\documentclass[a4paper,UKenglish,cleveref, autoref, thm-restate, numberwithinsect]{lipics-v2021}

\nolinenumbers
\hideLIPIcs

\title{A Parameterized-Complexity Framework for Finding Local Optima} 

\authorrunning{Robert Ganian, Hung P. Hoang,
Christian Komusiewicz,
Nils Morawietz}
\Copyright{Robert Ganian, Hung P. Hoang,
Christian Komusiewicz,
Nils Morawietz}

\author{Robert Ganian}{Algorithms and Complexity Group, TU Wien, Austria}{rganian@gmail.com}{https://orcid.org/0000-0002-7762-8045}{Austrian Science Foundation (FWF), project 10.55776/Y1329 and Vienna Science and Technology Fund (WWTF), project 10.47379/ICT22029}

\author{Hung P. Hoang}{Algorithms and Complexity Group, TU Wien, Austria}{phoang@ac.tuwien.ac.at}{https://orcid.org/0000-0001-7883-4134}{Austrian Science Foundation (FWF), projects 10.55776/Y1329 and ESP1136425}

\author{Christian Komusiewicz}{Institute of Computer Science, Friedrich Schiller University Jena, Germany}{c.komusiewicz@uni-jena.de}{https://orcid.org/0000-0003-0829-7032}{}

\author{Nils Morawietz}{LaBRI, Université de Bordeaux, France\\
Institute of Computer Science, Friedrich Schiller University Jena, Germany}{nils.morawietz@uni-jena.de}{https://orcid.org/0000-0002-7283-4982}{French ANR, project ANR-22-CE48-0001 (TEMPOGRAL).}


\category{} 
\relatedversion{} 

\acknowledgements{}

\EventEditors{}
\EventNoEds{0}
\EventLongTitle{}
\EventShortTitle{}
\EventAcronym{}
\EventYear{}
\EventDate{}
\EventLocation{}
\EventLogo{}
\SeriesVolume{}
\ArticleNo{}

\usepackage{amsfonts,amsmath,amssymb,amsthm,epsfig}
\usepackage{color}
\usepackage{graphicx} 
\usepackage{amsthm}
\usepackage{xspace}

\usepackage[shortlabels]{enumitem}
\usepackage{tabularx}
\usepackage{mdframed}
\usepackage{xpatch}
\usepackage{environ}
\usepackage{tikz}

\usetikzlibrary{calc}

\makeatletter
\newcommand{\problemtitle}[1]{\gdef\@problemtitle{#1}}\newcommand{\probleminput}[1]{\gdef\@probleminput{#1}}\newcommand{\problemtask}[1]{\gdef\@problemtask{#1}}\NewEnviron{problem}{
\problemtitle{}\probleminput{}\problemtask{}\BODY\noindent
\begin{tabularx}{\textwidth}{	 						l X c}
	\multicolumn{2}{						l}{\@problemtitle} \\	\textbf{Input:} & \@probleminput \\	\textbf{Task:} & \@problemtask\end{tabularx}
}
\makeatother

\makeatletter
\xpatchcmd{\endmdframed}
  {\aftergroup\endmdf@trivlist\color@endgroup}
  {\endmdf@trivlist\color@endgroup\@doendpe}
  {}{}
\makeatother

\surroundwithmdframed[linecolor=black,linewidth=1pt,skipabove=1mm,skipbelow=1mm,leftmargin=0mm,rightmargin=0mm,innerleftmargin=1mm,innerrightmargin=1mm,innertopmargin=1mm,innerbottommargin=1mm]{problem}

\newcommand{\Lup}{L^\uparrow}
\newcommand{\Ldo}{L^\downarrow}
\newcommand{\lup}{\ell^\uparrow}
\newcommand{\ldo}{\ell^\downarrow}
\newcommand{\Xup}{X^\uparrow}
\newcommand{\Xdo}{X^\downarrow}
\newcommand{\bigoh}{\ensuremath{\mathcal{O}}}

\newtheorem{mainresult}{Main Finding}

\newcommand{\PP}{\mathcal{P}}

\newcommand{\I}{\mathcal{I}}
\newcommand{\poly}[1]{{#1}^{\bigoh(1)}}
\newcommand{\Q}{\mathbb{Q}}
\newcommand{\Z}{\mathbb{Z}}

\newcommand{\pls}{\ensuremath{\mathsf{PLS}}\xspace}
\newcommand{\W}{\ensuremath{\mathsf{W}}}
\newcommand{\NP}{\ensuremath{\mathsf{NP}}}
\newcommand{\coNP}{\ensuremath{\mathsf{coNP}}\xspace}
\newcommand{\FPT}{\ensuremath{\mathsf{FPT}}}
\newcommand{\XP}{\ensuremath{\mathsf{XP}}\xspace}

\newcommand{\mr}{\mathcal{R}}
\newcommand{\Oh}{\mathcal{O}}
\newcommand{\wmax}{\omega_{\max}}

\newcommand{\plsredins}{\Phi}
\newcommand{\plsredsol}{\Psi}

\newcommand{\plsprob}[2]{{\textsc{#1/\allowbreak#2}}}
\newcommand{\pivotprob}[2]{\plsprob{#1}{#2}$[\textsc{Pivot}]$}

\newcommand{\sign}{\text{sign}}

\begin{document}

\maketitle

\begin{abstract}
Local search is a fundamental optimization technique that is both widely used in practice and deeply studied in theory, yet its computational complexity remains poorly understood. The traditional frameworks, \pls and the standard algorithm problem, introduced by Johnson, Papadimitriou, and Yannakakis (1988) fail to capture the methodology of local search algorithms: \pls is concerned with finding a local optimum and not with using local search, while the standard algorithm problem restricts each improvement step to follow a fixed pivoting rule.
In this work, we introduce a novel formulation of local search
which provides a middle ground between these models.
In particular, the task is to output not only a local optimum but also a chain of local improvements leading to it.
With this framework, we aim to capture the challenge in designing a good pivoting rule.
Especially, when combined with the parameterized complexity paradigm, it enables both strong lower bounds and meaningful tractability results.
Unlike previous works that combined parameterized complexity with local search, our framework targets the whole task of finding a local optimum and not only a single improvement step. Focusing on two representative meta-problems---\textsc{Subset Weight Optimization Problem} with the $c$-swap neighborhood and \textsc{Weighted Circuit} with the flip neighborhood---we establish fixed-parameter tractability results related to the number of distinct weights, while ruling out an analogous result when parameterized by the distance to the nearest optimum via a new type of reduction. 
\end{abstract}

\section{Introduction}

Local search is one of the most commonly employed paradigms in the design of algorithms for optimization problems. The idea underlying local search algorithms is both natural and simple: start from an initial solution $S$ to the problem and apply a series of local improvement steps until reaching a local optimum, that is, a solution which cannot be improved by any local change. The definition of local improvement steps typically involves straightforward operations such as swapping or moving a constant number of elements in the solution. There are almost always many ways these steps could be applied, and one typically speaks of the \emph{local neighborhood} of a solution $S$ to subsume all possible solutions that can be obtained from $S$ by performing one round of permissible operations. A solution~$S$ is called a \emph{local optimum} if there is no better solution in its local neighborhood, and local search algorithms are typically tasked with finding such a local optimum via a sequence of local improvement steps.

While the local search paradigm often performs very well in practice, currently established complexity-theoretic tools exclude tractability of finding local optima for many problems of interest~\cite{aarts97,MR2450938}.
The central complexity class in the theory of local search is \pls~\cite{JPY1988} 
(for ``\emph{polynomial local search}''). Inclusion in \pls essentially means that one can search the local neighborhood of any given solution in polynomial time\footnote{Formal definitions are provided in Section~\ref{sec:prelims}.}. While inclusion in \pls is a natural prerequisite for efficient local search, it does not guarantee that one can find a local optimum in polynomial time---to the contrary, it is now widely believed that \pls-complete problems do not admit any polynomial-time local search algorithm~\cite{MR2450938}.
In this sense, establishing \pls-hardness for a local search problem can be seen as a counterpart to establishing \textsf{NP}-hardness for decision problems.

Given the above discrepancy between the performance of local search algorithms and inclusion in \pls, researchers have proposed a second classical formalization of local search: the \emph{standard algorithm problem}~\cite{JPY1988}. There, one fixes not only the notion of local neighborhood but also the specific \emph{pivoting rule}, that is, the specific procedure used to compute the next solution in the local improvement step. The underlying task is then to identify the local optimum that will be obtained by exhaustively applying the selected pivoting rule from a provided initial solution $S$. The standard algorithm problem provides a useful way of making complexity-theoretic statements about local search---several standard algorithm problems are known to be \textsf{PSPACE}-complete~\cite{schaffer1991,Papadimitriou92,Schulman00,ChenGVYZ20,MantheyMRS24}---but suffers from the drawback that each such statement applies to a combination of pivoting rule and local search problem. This drawback has already been acknowledged in the seminal work of Schäffer and Yannakakis~\cite[Page 62, 1st paragraph]{schaffer1991}: while formal lower bounds made for the standard algorithm problem only hold for specific pivoting rules (such as, e.g., moving to the first discovered improvement in the local neighborhood), it is more desirable to exclude efficient local search for \emph{any} pivoting rule and this is also what can be inferred from most existing lower bound constructions~\cite{JPY1988,Papadimitriou92,Schulman00,ChenGVYZ20,MantheyMRS24}.

Crucially, regardless of whether one chooses the \pls or standard algorithm formulation, most local search problems remain computationally intractable. 
In this article, we introduce a novel perspective which allows us to establish not only stronger \textbf{lower bounds}, but also \textbf{tractability} for local search problems. We do so by developing a theory of parameterized local search which is inspired by and connected to the well-established \emph{parameterized complexity} paradigm~\cite{DowneyF13,CyganFKLMPPS15}. What distinguishes our work from prior studies that examined parameterized complexity of local search~\cite{MarxS10,Szeider11,KrokhinM12,GuoHNS13,GarvardtGKM23,KomusiewiczM22,GarvardtMNW23} is that our results apply to the local search task as a whole and not only to a single improving step (see the discussion of related work at the end of this section). In fact, our framework targets a wide variety of common settings where computing a single improving step is polynomial-time solvable. Moreover, as a byproduct of our approach we also obtain a formulation of local search problems that lies between \pls and the standard algorithm: one where negative results exclude tractability with respect to every pivoting rule (unlike the latter) while positive ones guarantee efficient local search (unlike the former).

\subparagraph{Contributions.}
We develop our framework on two types of problems arising in the local search setting:
\begin{enumerate}
\item \textsc{Subset Weight Optimization Problem/$c$-Swap}~\cite{KomusiewiczM25}, whose underlying optimization problem generalizes weighted variants of many graph problems such as \textsc{Maximum Independent Set}, \textsc{Max Cut}, \textsc{Minimum Dominating Set}, \textsc{Maximum Vertex Cover}, and others. Here, solutions are vertex or edge subsets satisfying some arbitrary certifiable property and the local neighborhood is defined by swapping up to $c$ vertices or edges for some fixed constant~$c$.
\item \textsc{Weighted Circuit/Flip}, which is the joint generalization of the classical \textsc{Max Circuit/Flip} and \textsc{Min Circuit/Flip}\ problems~\cite{JPY1988,aarts97} 
to arbitrary weights. Here, the input is a circuit along with a weight function over the output gates, solutions are input bit vectors and the local neighborhood is defined by flipping a single bit on the input, that is, inverting its bit.
\end{enumerate}

For all the above problems, we primarily target the following local search task: given an initial solution $S$, output a series of improving steps from $S$ to a local optimum. 
We note that this \emph{pivoting} formulation forms a middle ground that avoids the drawbacks of both previously studied formulations of local search (see~\Cref{fig:comparison} for an illustration): 
\begin{itemize}
\item In the \emph{\pls} formulation, a hypothetical algorithm can find the local optimum in an arbitrary way (without requiring improvement steps and potentially even yielding a solution worse than the given starting solution);
\item In the standard algorithm formulation, one is forced to use a fixed pivoting rule and all obtained complexity-theoretic statements apply only to that specific rule.
\end{itemize}

\begin{figure}
	\centering
	\includegraphics[width=\textwidth]{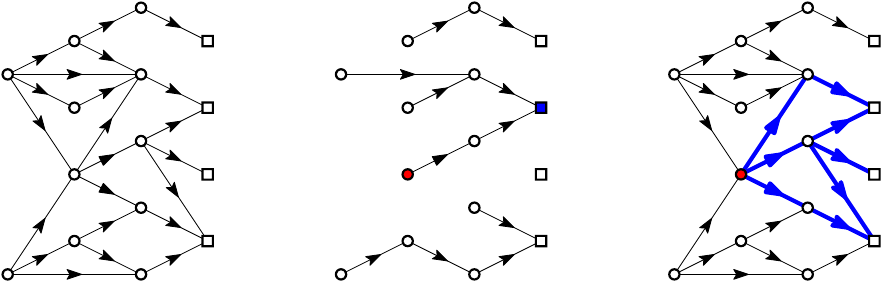}
	\caption{Comparison of the three models. The graphs are transition graphs, where solutions are vertices and local improvements are edges.	
	\textbf{Left:} The \pls formulation asks for a local optimum (i.e., any square in the transition graph) without needing to follow a sequence of improvement steps. \textbf{Middle:} In the standard algorithm problem, we have a specific pivoting rule (i.e., every solution has at most one outgoing edge) and a given initial solution (depicted as the red circle). The task is to find the unique local optimum (blue square) reachable from the initial solution. \textbf{Right:} In our pivoting formulation, given an initial solution (red circle), we want to output an improving sequence from that solution to a local optimum (i.e., any maximal path in the blue subgraph).}
	\label{fig:comparison}
\end{figure}

As a direct consequence of previously established lower bounds~\cite{schaffer1991,MonienT10,abs-2310-19594}, the problems mentioned above cannot admit a polynomial-time algorithm in the pivoting formulation. Moreover, such results are \emph{unconditional}: there simply exist instances with initial solutions $S$ for which the shortest sequence of local improvement steps is exponential (the so-called \emph{all-exp} property~\cite{schaffer1991}), and hence a local search algorithm cannot terminate in polynomial time. However, this does not rule out tractability in the parameterized sense---there, one typically asks for so-called \emph{fixed-parameter tractability}, which means that there is an algorithm terminating in time $f(k)\cdot n^{\bigoh(1)}$ for input size $n$ and some computable function $f$ of a (possibly promised) parameter~$k$.

After setting up the conceptual contributions outlined above, we proceed to our technical contributions: a parameterized analysis of local search problems. We summarize our two main findings below.

\begin{mainresult}
\label{main:positive}
Local search for \textsc{Subset Weight Optimization Problem/$c$-Swap} and \textsc{Weighted Circuit/Flip} is fixed-parameter tractable when parameterized by the number of distinct weights.
\end{mainresult}

\begin{mainresult}
\label{main:negative}
Unless $\FPT = \W[1]$, neither \textsc{Subset Weight Optimization Problem/$c$-Swap} nor \textsc{Weighted Circuit/Flip} admit a fixed-parameter local search algorithm when parameterized by the distance (measured by number of improvement steps) to the nearest local optimum.
\end{mainresult}

Main Finding~\ref{main:positive} holds for all considered formulations of local search (i.e., \pls, pivoting, standard algorithm), and is primarily based on Theorems~\ref{thm:diff_weights}-\ref{thm:diff_weights_circuit}. While these proofs are not technically challenging and rely on the weight reduction technique by Frank and Tardos~\cite{MR905151}, the finding showcases that parameterized complexity can provide meaningful tractability results for local search. 
In particular, it provides a bridge between the polynomial runtime of local search for many unweighted problems and the hardness results of their corresponding weighted variants.

Main Finding~\ref{main:negative} applies only to the newly proposed pivoting formulation of local search. Indeed, on one hand the statement is vacuous in the standard algorithm formulation (if an algorithm cannot choose which improvement step to make, the distance to the local optimum is irrelevant). 
On the other hand, the existence of a \emph{direct} connection between \pls-hardness and a collapse of complexity-theoretic decision classes is a long-standing open problem in the field---and while we do make progress towards establishing such a connection, completely settling this is beyond our current understanding (see also the Concluding Remarks in Section~\ref{sec:concl}).
To the best of our knowledge, Main Finding~\ref{main:negative} is also the first lower bound establishing that the hardness of local search is neither due to the non-existence of a nearby local solution (i.e., the all-exp property)~\cite{PapadimitriouSY90,schaffer1991,MonienDT10,Heimann0H24}
 nor due to the intractability of computing an improvement step~\cite{MarxS10,Szeider11,KrokhinM12,GuoHNS13,GarvardtGKM23,KomusiewiczM22,GarvardtMNW23}: 
 already finding the ``correct'' improvement steps is hard.

Main Finding~\ref{main:negative} is based on \cref{thm:wis_hard} and \cref{cor:wis_hard_circuit}, which are both non-trivial. 
As a starting point towards the former, we show that \plsprob{Maximum Independent Set}{3-Swap}, a special case of \plsprob{Subset Weight Optimization Problem}{$c$-Swap}, has the \emph{all-exp property}.
This means that for infinitely many pairs of an instance and an initial solution, all reachable local optima are exponentially far away from the initial solution; that is, regardless of any pivoting rule, the standard local search algorithm always runs in exponential time.
We establish this property via a tight \pls-reduction from \plsprob{Max Cut}{Flip}.\footnote{Here, the local neighborhood allows to change for any single vertex the side of side of the partition. For example, the partition~$(A,B)$ has the flip-neighbor~$(A\setminus\{v\},B\cup\{v\})$ for each~$v\in A$.}
This notion of reduction is stronger than the original \pls-reduction and can be used to transfer the all-exp property; see \cref{sec:prelims} for formal definitions of these reductions.
While there already exists a \pls-reduction from \plsprob{Max Cut}{Flip} to \plsprob{Maximum Independent Set}{3-Swap}~\cite{KomusiewiczM25}, that reduction is not tight and our adaptations to the known reduction require an involved analysis to show tightness.
Equipped with this all-exp property, we are able to construct a reduction from the canonical \W[1]-hard problem \textsc{Multicolored Independence Set} to \plsprob{Maximum Independent Set}{3-Swap}, thus establishing Main Finding~\ref{main:negative} for the latter problem.

Towards establishing Main Finding~\ref{main:negative} for \plsprob{Weighted Circuit}{Flip}, we do not start from a known \W[1]-hard decision problem (such as \textsc{Multicolored Independence Set}) but instead develop a new notion of reduction which can translate parameterized local search lower bounds directly.
This notion adds to the well-established tight \pls-reduction a constraint that guarantees parametric stability by forbidding the creation of new long paths in the transition graph. Our result here---a concrete application of such a reduction from \plsprob{Maximum Independent Set}{3-Swap} to \plsprob{Max-Circuit}{Flip}---can thus be seen as a demonstration that the new parameterized lower bounds can be translated to pivoting formulations of other local search problems of interest.

\subparagraph{Related Work.}
The \pls\ and standard algorithm formulations have been studied for a broad variety of search problems, including not only those captured by \textsc{Subset Weight Optimization Problem/$c$-Swap} and \textsc{Weighted Circuit/Flip} but also the \textsc{Simplex Method}~\cite{FearnleyS15} and \textsc{Gradient Descent}~\cite{FearnleyGHS23}. The all-exp property (which we newly establish for \plsprob{Weighted Independent Set}{3-Swap}) has previously been conjectured and proven for \textsc{Traveling Salesman/$k$-Opt}~\cite{Krentel89,Heimann0H24} and \plsprob{Max Cut}{Flip}~\cite{MonienT10,KomusiewiczM25}, among others. Whether the \textsc{Simplex Method} admits a pivoting rule that guarantees a polynomial number of iterations is one of the most important open problems in the area of linear programming.

The parameterized complexity of local search was studied for a large number of problems, including: \textsc{Stable Marriage/$k$-Swap}~\cite{MarxS10}, \textsc{Max SAT/$k$-Flip}~\cite{Szeider11}, \textsc{Max CSP/$k$-Flip}~\cite{KrokhinM12}, \textsc{Traveling Salesman} under a variety of local distance measures~\cite{GuoHNS13}, \textsc{Max $c$-Cut/$k$-Flip}~\cite{GarvardtGKM23} and \textsc{Cluster Editing/$k$-Move}, \textsc{Cluster Deletion/$k$-Move}~\cite{GarvardtMNW23}.
In all of the above studies, the parameter measured the complexity of finding a single improving step in the local neighborhood. That perspective is complementary to the one explored in this paper---we target the practically motivated case where performing a single improving step is easy, yet finding a local optimum is difficult. 

\subparagraph{Paper Organization.}
We begin by setting up the basic preliminaries in Section~\ref{sec:prelims} and introducing our framework in Section~\ref{sec:framework}. Then, we proceed with establishing the technically involved lower bounds captured by Main Finding~\ref{main:negative} for \plsprob{Subset Weight Optimization Problem}{$c$-Swap} in Section~\ref{sec:hardness} and for \plsprob{Weighted Circuit}{Flip} in Section~\ref{sec:newred}. The positive results described in Main Finding~\ref{main:positive} are detailed in Section~\ref{sec:fptalgo}, and we conclude with a discussion of possible future directions and open questions in Section~\ref{sec:concl}.

\section{Preliminaries}
\label{sec:prelims}

For two integers $a < b$, we denote by $[a,b]$ the set of integers between (and including) $a$~and~$b$.
We write $[a]$ for $[1,a]$.
Given a set $S$, a subset $S'$ of $S$, and a function $f : S \to \Q$, we define $f(S') := \sum_{s \in S'} f(s)$.
For two sets $A$ and $B$, we denote by $A \oplus B$ their symmetric difference (i.e., $A \oplus B := (A \setminus B) \cup (B \setminus A))$.
Given a graph $G$, the open neighborhood~$N_G(v)$ of a vertex $v$ is the set of all adjacent vertices of $v$ in $G$.
Its closed neighborhood~$N_G[v]$ is defined as $N_G(v) \cup \{v\}$.
We drop the subscript when the graph $G$ is clear.
For a vertex subset~$S \subseteq V(G)$, we denote by~$G[S]$ the subgraph of~$G$ induced by~$S$.

A \emph{local search problem} $P$ is an optimization problem that consists of a set of instances~$D_{P}$, a finite set of (feasible) solutions $F_{P}(I)$ for each instance $I\in D_{P}$, an objective function $f_{P}$ that assigns an integer value to each instance $I\in D_{P}$ and solution $s\in F_P(I)$, and a neighborhood $N_{P}(s,I)\subseteq F_{P}(I)$ for each solution $s\in F_{P}(I)$. 
The size of every solution $s \in F_{P}(I)$ is bounded by a polynomial in the size of $I$. 
The goal is to find a \emph{locally optimal solution} for a given instance $I$; that is, a solution $s \in F_{P}(I)$ such that no solution $s' \in N_{P}(s,I)$ yields a better objective value than $f_P(s,I)$.
Formally, this means that for all $s'\in N_{P}(s,I)$, we have~$f_{P}(s,I)\leq f_{P}(s',I)$ if $P$ is a minimization problem, and $f_{P}(s,I)\geq f_{P}(s',I)$ if $P$ is a maximization problem.

A common naming convention for a local search problem is of the form $Q/N$, where $Q$ is the corresponding optimization problem and $N$ is the notion of neighborhood.
For example, \textsc{Max Cut/Flip} indicates the local search problem with instances, solutions, and objective function as defined by the optimization problem \textsc{Max Cut} with the addition of the neighborhood as defined by the flip operation.

A \emph{standard local search algorithm} for an instance $I$ proceeds as follows.
It starts with some initial solution $s \in F_{P}(I)$.
Then it iteratively visits a neighbor with better objective value, until it reaches a local optimum. 
If a solution has more than one better neighbor, the algorithm has to choose one by some prespecified rule, often referred to as a \emph{pivoting rule}.
The \emph{standard algorithm problem} is then defined as follows~\cite{JPY1988}: Given an instance $I$ and an initial solution $s$, output the local optimum that would be produced by the standard local search algorithm (with a specific pivoting rule) starting from $s$.

Next, we formalize a series of classical notions related to the complexity class \pls.

\begin{definition}[The class~\pls~\cite{JPY1988}]
	A local search problem $P$ is in the class \emph{\pls}, if there are three polynomial time algorithms $A_{P}, \ B_{P}, \ C_{P}$ such that 
	\begin{itemize}
		\item Given an instance $I \in D_P$, $A_{P}$ returns a solution  $s \in F_{P}(I)$;
		\item Given an instance $I \in D_{P}$ and a solution $s\in  F_{P}(I)$, $B_{P}$ computes the objective value $f_{P}(s,I)$ of $s$; and 
		\item Given an instance $I \in D_P$ and a solution $s \in F_{P}(I)$, $C_{P}$ returns a neighbor of $s$ with strictly better objective value, if it exists, and ``locally optimal'', otherwise.
	\end{itemize}
\end{definition}

The ``basic'' reductions used for \pls are defined below.
\begin{definition}[\pls-reduction~\cite{JPY1988}]
\label{def:pls_complete}
	A \emph{\pls-reduction} from a local search problem $P$ to a local search problem $Q$ is a pair of polynomial-time computable functions $\plsredins$ and $\plsredsol$ that satisfy:
	\begin{enumerate}
		\item Given an instance $I \in D_{P}$, $\plsredins$ computes an instance $\plsredins(I) \in D_{Q}$; and
		\item Given an instance $I \in D_{P}$ and a solution $s_q \in F_{Q}(\plsredins(I))$, $\plsredsol$ returns a solution $s_p \in F_{P}(I)$ such that if $s_q$ is a local optimum for $\plsredins(I)$, then $s_p$ is a local optimum for $I$. 
	\end{enumerate} 
	A problem $Q\in \pls$ is \emph{\pls-complete} if for every problem~$P \in \pls$, there exists a \pls-reduction from $P$ to $Q$.
\end{definition}

A more restrictive (and often useful) kind of reduction was defined later, by Sch{\"a}ffer and Yannakakis. Before defining it, we will need the following notion:

\begin{definition}[Transition graph]
	The \emph{transition graph} $T_I$ of an instance $I$ of a local search problem is a directed graph such that the vertices are solutions of $I$, and an edge $(s, s')$ exists if $s'$ is a neighbor of $s$ with a better objective value.
\end{definition}

\begin{definition}[Tight \pls-reduction~\cite{schaffer1991}]
\label{def:tight_pls_complete}
	A \emph{tight \pls-reduction} from a local search problem $P$ to a local search problem $Q$ is a \pls-reduction $(\plsredins,\plsredsol)$ from $P$ to $Q$ such that for every instance $I \in D_P$, we can choose a subset $R$ of $F_Q(\plsredins(I))$ that satisfies:
	\begin{enumerate}
		\item $R$ contains all local optima of $\plsredins(I)$;
		\item For every solution $s_p \in F_{P}(I)$, we can construct in polynomial time a solution $s_q \in R$ so that $\plsredsol(I,s_q) = s_p$;
		\item If there is a path from $s_q \in R$ to $s'_q \in R$ in the transition graph $T_{\plsredins(I)}$ such that there are no other solutions in $R$ on the path, then either there is an edge from the solution $\plsredsol(I,s_q)$ to the solution $\plsredsol(I,s'_q)$ in the transition graph $T_I$ or these two solutions are identical.
	\end{enumerate} 
	A problem $Q\in \pls$ is \emph{tightly \pls-complete} if for every problem~$P \in \pls$, there exists a tight \pls-reduction from $P$ to $Q$.
\end{definition}

Note that tight \pls-reductions preserve the \textsf{PSPACE}-completeness of the standard algorithm problem and also the all-exp property (formalized below)~\cite{schaffer1991}.

\begin{definition}[All-exp property]
    A local search problem~$P$ has the \emph{all-exp} property if there are infinitely many pairs of an instance~$I$ of $D_P$ and an initial solution $s \in F_P(I)$, for which the standard local search algorithm always needs an exponential number of iterations for all possible pivoting rules. 
\end{definition}

\subparagraph{Parameterized Complexity.} 
In parameterized complexity~\cite{DowneyF13,CyganFKLMPPS15}, the complexity of a problem is studied not only with respect to the input size, but also with respect to some problem parameter(s). 
The input to a parameterized problem is a tuple $(\omega,\kappa)$ where $\omega$ is a string in a fixed alphabet and $\kappa$ is a non-negative integer called the \emph{parameter}. Typically, $\kappa$ captures some property of a sought-after solution (e.g., an upper bound on the solution size), or a promised property of $\omega$ (e.g., an upper bound on the treewidth or clique-width of the input graph). A parameterized problem is \emph{fixed-parameter tractable} if all instances where the promised property holds can be solved in time $f(\kappa)\cdot |\omega|^{\bigoh(1)}$ for some computable function $f$. 

\begin{remark}
The above promise-based formulation can be avoided if one assumes that a witness for the structure captured by $\kappa$ is provided on the input~\cite[page 137]{CyganFKLMPPS15}; however, this is unreasonable if we wish to parameterize by the distance to the nearest optimal solution. Alternatively, one may formulate the parameter as a function $\zeta$ of $\omega$~\cite{FlumG06}, and all our results may equivalently be stated in this formulation---but only if one does not place restrictions on the running time required to compute $\zeta$ (see~\cite{Chandoo18} for a discussion of this type of parameterization). 
\end{remark}

A well-established complexity assumption that we use for our lower bounds is that the class $\FPT$ of all fixed-parameter tractable parameterized decision problems is not equal to the class $\W[1]$. In other words, it is considered unlikely for a $\W[1]$-hard problem to be fixed-parameter tractable, and the existence of such an algorithm would---among others---violate the Exponential Time Hypothesis~\cite[Section 14.4]{CyganFKLMPPS15}.

\section{Setting up the Framework}
\label{sec:framework}
We start with the concept of pivoting formulation.
Given an instance $I$ of a local search problem, we define an \emph{improving sequence} for $I$ to be a sequence of solutions $s_0, \dots, s_r$ such that for $i \in [r]$, $s_i$ is a neighbor of $s_{i-1}$ with a better objective value.
An improving sequence is \emph{maximal} if the last solution is locally optimal.

\begin{definition}
	For a local search problem $\PP$, the corresponding pivoting problem $\PP[\textsc{Pivot}]$ is defined as follows: Given an instance $I$ of $\PP$ and an initial solution $s$ for $I$, output a maximal improving sequence for $I$ starting from $s$.
\end{definition}

Observe that for a local search problem, the task is to find a local optimum, without the restriction on the techniques used to achieve it.
For the pivoting problem, we restrict ourselves to using only a local search algorithm (i.e., following a path in the transition graph).
We can view it as finding the right pivoting rule to arrive at a local optimum, and hence the name ``pivoting''.

For example, consider the following local search problem \plsprob{Subset Weight Optimization Problem}{$c$-Swap}: We are given a graph $G=(V,E)$, a weight function $\omega: {V \cup E} \to \Q$, a \emph{certifying function} $\sigma : 2^{V \cup E} \to \{0,1\}$ that certifies whether a subset of vertices and edges form a solution, and an objective function $f : 2^{V \cup E} \to \Q$ with~$f(S) := \sum_{u\in S}\omega(u)$.
Two solutions are neighbors if they differ by a \emph{$c$-swap} (i.e., their symmetric difference is at most~$c$).
The task is to find a locally maximal solution (i.e., a solution $S \subseteq V \cup E$ such that for all neighboring solution $S'$ of $S$, we have $f(S) \geq f(S')$).
The corresponding pivoting problem is then defined as follows.

\begin{problem}
	\problemtitle{\pivotprob{Subset Weight Optimization Problem}{$c$-Swap}}
	\probleminput{A graph $G=(V,E)$, a weight function $\omega: {V \cup E} \to \Q$, a certifying function $\sigma: 2^{V \cup E} \to \{0,1\}$, and an initial solution $S$ (i.e., $\sigma(S) = 1$).}
	\problemtask{Output a maximal improving sequence starting from $S$.}
\end{problem}

We assume that $\omega$ and~$\sigma$ can be computed in time polynomial in $|V|$.
Note that the graph $G$ may be directed and/or a multi-graph.
Further, although it is phrased here as a maximization problem, by reversing the signs of the images of $\omega$, we can also model a minimization problem.
We also assume $c$ to be a constant so that the local search variant is in \pls (i.e., a polynomial time algorithm to compute a better neighbor exists).

\plsprob{Subset Weight Optimization Problem}{$c$-Swap} can be considered as a general problem that takes as special cases many well-known local search problems, such as the following.

\begin{itemize}
	\item \plsprob{Weighted Independent Set}{$c$-Swap} (resp., \plsprob{Weighted Clique}{$c$-Swap}, \plsprob{Weighted Vertex Cover}{$c$-Swap}):  In this case, the graph is vertex-weighted (i.e., the edge weights are zero); and a solution is an independent set (resp., a clique, a vertex cover).
	\item \plsprob{Traveling Salesman}{$k$-opt}: We have an edge-weighted graph, edge sets of spanning cycles are the solutions; and $c = 2k$ (i.e., a $k$-opt step is a $(2k)$-swap).
	\item \plsprob{Max Cut}{Flip} with bounded maximum degree $\Delta$: Here, the graph is edge-weighted (i.e., the vertex weights are zero); a solution is a set of vertices~$U$ together with all edges between~$U$ and~$V\setminus U$; and the swap size is~$c = 2\cdot\Delta + 1$.
\end{itemize}
\begin{remark}
To model exactly the flip-neighborhood of \plsprob{Max Cut}{Flip}, we add for each vertex~$v\in V$ additionally~$\Delta$ isolated vertices~$v_i$ with~$i\in [1,\Delta]$ and enforce that each valid solution has to contain either all or none of $\{v\}\cup \{v_i\mid 1\leq i \leq \Delta\}$ for each vertex~$v\in V$.
This ensures that the swap size of~$c=2\cdot \Delta + 1$ does not allow us to swap more than one vertex together with its associated isolated vertices. 
Thus, two neighboring solutions differ by at most one original vertex, which then allows to model exactly the flips of single vertices in \plsprob{Max Cut}{Flip}.
\end{remark}

The next local search problem we consider is \plsprob{Weighted Circuit}{Flip}, a generalization of the classical \plsprob{Circuit}{Flip}.
In this problem, we are given a Boolean circuit $D$ with $n$ input nodes $x_1, \dots, x_n$ and $m$ output nodes $y_1, \dots, y_m$ and a weight function $\omega: [m] \to \Q$. 
The goal is to find a locally maximal input for the optimization function $f(x_1, \dots, x_n) = \sum^m_{i=1} \omega(i) y_i$ and 1-swap neighborhood (which we also refer to as a \emph{flip}).
The corresponding pivoting problem is as follows.

\begin{problem}
	\problemtitle{\pivotprob{Weighted Circuit}{Flip}}
	\probleminput{A Boolean circuit $D$ with $n$ input nodes $x_1, \dots, x_n$ and $m$ output nodes $y_1, \dots, y_m$, a weight function $\omega: [n] \to \Q$, and an initial solution $s \in \{0,1\}^n$.
}	
	\problemtask{Output a maximal improving sequence starting from $s$.}
\end{problem}

Note that when $\omega(i) = 2^{i-1}$ (i.e., we interpret $f(x_1, \dots, x_n)$ as the number $y_m \dots y_1$ written in binary), then we recover the classical \plsprob{Max-Circuit}{Flip}.
Similarly, if we set $\omega(i) = -2^{i-1}$, then we have the \plsprob{Min-Circuit}{Flip}.
Further, this problem also captures \plsprob{Weighted Max-SAT}{Flip}.

In our parameterized analysis of the above problems, we assume that the input is additionally equipped with an integer parameter~$\kappa$ (see Section~\ref{sec:prelims}).

\section{Hardness results for \plsprob{Weighted Independent Set}{3-Swap}}
\label{sec:hardness}
Recall that the parameter~$\ell$ is the distance of the given starting solution to the nearest local optimum in the transition graph.
Observe that the time complexity of a pivoting problem is lower-bounded by the encoding length of the shortest maximal improving sequence from the initial solution (i.e., $\ell$ times the encoding size of a solution).
Hence, it is natural to consider the parameterized complexity of such a problem with respect to $\ell$.

This section is dedicated to prove the following theorem.
\begin{theorem}
\label{thm:wis_hard}
Unless $\FPT = \W[1]$, there does not exist an algorithm to solve \pivotprob{Weighted Independent Set}{3-Swap} in \FPT-time when parameterized by $\ell$ (i.e., in time $g(\ell) \cdot \poly{n}$ for some computable function $g$).
\end{theorem}

Note that since \pivotprob{Weighted Independent Set}{3-Swap} is a special case of \pivotprob{Subset Weight Optimization Problem}{$c$-Swap}, the hardness result above also extends to the latter problem.

\begin{remark}
	\cref{thm:wis_hard} is tight in the sense that for any instance of \pivotprob{Subset Weight Optimization Problem}{$2$-Swap} with all positive weights or all negative weights, the output always has length $\Oh(n^2)$.
	Indeed, suppose that all weights are positive; the argument for negative weights is analogous.
	We label the vertices as $v_1, \dots, v_n$ in the increasing order of their weights.
	For every solution $R$, define its potential $p(R) := \sum_{v_i \in R} i$.
	Then it is easy to see that every 2-swap increases the potential.
	Since the potential only has $\Oh(n^2)$ possible values, every improving sequence must have length $\Oh(n^2)$ as well.
\end{remark}

The proof of \cref{thm:wis_hard} makes use of the following result. 

\begin{theorem}
\label{thm:wis_all_exp}
There is a tight \pls-reduction from \plsprob{Max-Cut}{Flip} with bounded degree to \plsprob{Weighted Independent Set}{3-Swap}. In particular, the latter problem has the all-exp property (even when restricting to positive weights).
\end{theorem}

\plsprob{Weighted Independent Set}{3-Swap} has been shown to be \pls-complete~\cite{KomusiewiczM25} via a \pls-reduction from \plsprob{Max Cut}{Flip}.
However, this reduction was not argued to be tight, and is in fact unlikely to be.
Here, we adapt it into a tight \pls-reduction; the proof is presented in Section~\ref{subsec:wis_all_exp}.

The constructions in the proofs of Theorems~\ref{thm:wis_hard} and \ref{thm:wis_all_exp} use the concept of (up/down)-elevators, defined as follows.
Here, when we assume a certain (fixed) ordering on a set $X$, we denote by $X[i]$ the $i$th element in this ordering.

Let~$G=(V,E)$ be a vertex-weighted graph and let~$X \subseteq V$ be a set of size at least~$2$.
An~\emph{$X$-elevator} is a clique~$L := \{\ell_1, \dots, \ell_{s}\}$ of size~$s := |X|-1$, such that the vertices of~$L$ have the same neighborhood outside of~$L \cup X$ and for each~$i\in [1,s-1]$, $\ell_i$ has exactly the first~$i+1$ vertices of~$X$ as neighbors inside of~$X$ for an arbitrary but fixed ordering on the vertices of~$X$.
We call the vertices of an elevator~\emph{levels}.
The vertex~$\ell_{|L|}$ is the~\emph{top-level}, while $\ell_1$ is the~\emph{bottom-level}.

We say that~$L$ is an~\emph{up-elevator} if~$\omega(\ell_1) = \omega(X[1]) + \omega(X[2]) + 1$ and for each~$i\in[2,s-1]$, $\omega(\ell_i) = \omega(\ell_{i-1}) + \omega(X[i+1]) + 1$.
Similarly, we say that~$L$ is a~\emph{down-elevator} if~$\omega(\ell_1) = \omega(x_1) + \omega(x_2) -1$ and for each~$i\in[2,s-1]$, $\omega(\ell_i) = \omega(\ell_{i-1}) + \omega(X[i+1]) -1$. 
Note that the weights of the vertices in such an up- or down-elevator are uniquely defined by the weights and the order of the vertices of~$X$.
See~\Cref{elevator examples} for examples of an up-elevator and a down-elevator.

\begin{figure}
\centering
\begin{tikzpicture}
\tikzstyle{knoten}=[circle,fill=white,draw=black,minimum size=5pt,inner sep=0pt]

\draw[rounded corners] (.7,-.5) rectangle (4.3, .5) {};

\node[knoten] (x1) at (1,0) {};
\node[knoten] (x2) at (2,0) {};
\node[knoten] (x3) at (3,0) {};
\node[knoten] (x4) at (4,0) {};

\begin{scope}[yshift=.5cm]

\node[knoten] (l1) at (2,1) {};
\node[knoten] (l2) at (3,1.5) {};
\node[knoten] (l3) at (4,2) {};

\draw[rounded corners] (1.7,.5) rectangle (4.3, 2.5) {};
\end{scope}

\draw[-] (x1) to (l1);
\draw[-] (x2) to (l1);
\draw[-] (x1) to (l2);
\draw[-] (x2) to (l2);
\draw[-] (x1) to (l3);
\draw[-] (x2) to (l3);
\draw[-] (x3) to (l2);
\draw[-] (x3) to (l3);
\draw[-] (x4) to (l3);

\node () at ($(x1)-(0,.3)$) {2};
\node () at ($(x2)-(0,.3)$) {5};
\node () at ($(x3)-(0,.3)$) {9};
\node () at ($(x4)-(0,.3)$) {1};

\node () at ($(l1)+(0,.3)$) {8};
\node () at ($(l2)+(0,.3)$) {18};
\node () at ($(l3)+(0,.3)$) {20};
\node () at ($(l3)-(2,-.25)$) {$L$};
\node () at ($(x1)+(0,.8)$) {$X$};

\begin{scope}[xshift=6cm]

\draw[rounded corners] (.7,-.5) rectangle (4.3, .5) {};

\node[knoten] (x1) at (4,0) {};
\node[knoten] (x2) at (3,0) {};
\node[knoten] (x3) at (2,0) {};
\node[knoten] (x4) at (1,0) {};

\begin{scope}[yshift=.5cm]

\node[knoten] (l1) at (3,1) {};
\node[knoten] (l2) at (2,1.5) {};
\node[knoten] (l3) at (1,2) {};

\draw[rounded corners] (.7,.5) rectangle (3.3, 2.5) {};
\end{scope}

\draw[-] (x1) to (l1);
\draw[-] (x2) to (l1);
\draw[-] (x1) to (l2);
\draw[-] (x2) to (l2);
\draw[-] (x1) to (l3);
\draw[-] (x2) to (l3);
\draw[-] (x3) to (l2);
\draw[-] (x3) to (l3);
\draw[-] (x4) to (l3);

\node () at ($(x1)-(0,.3)$) {3};
\node () at ($(x2)-(0,.3)$) {4};
\node () at ($(x3)-(0,.3)$) {2};
\node () at ($(x4)-(0,.3)$) {9};

\node () at ($(l1)+(0,.3)$) {6};
\node () at ($(l2)+(0,.3)$) {7};
\node () at ($(l3)+(0,.3)$) {15};

\node () at ($(l3)+(2,.25)$) {$L'$};
\node () at ($(x1)+(0,.8)$) {$X'$};

\end{scope}
\end{tikzpicture}
\caption{An example for the elevator gadgets. On the left an $X$-up-elevator~$L$ and on the right an $X'$-down-elevator~$L'$. The edges of the cliques~$L$ and~$L'$ are not depicted. Moreover, the edges between the vertices of~$X$ or the vertices of~$X'$ are arbitrary.
Recall that the vertices of~$L$ (respectively~$L'$) have the same neighborhood outside of~$L\cup X$ (respectively~$L'\cup X'$).}
\label{elevator examples}
\end{figure}
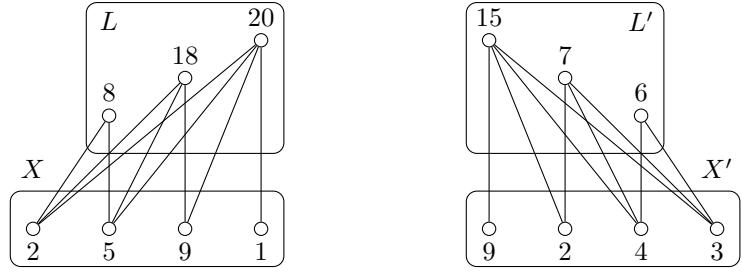

\begin{lemma}
\label{lem:elevator_impr}
Let $(G, \omega)$ be a weighted graph.
Let $L = \{\ell_1, \dots, \ell_s\}$ be an $X$-up-elevator for some vertex set $X$ of $G$, where $\ell_s$ is the top-level.
If an independent set~$S$ contains $\ell_i$ for some $i \neq s$, then $(S \setminus \{\ell_i, X[i+2]\}) \cup \{\ell_{i+1}\}$ is an independent set with a higher weight than that of $S$.
\end{lemma}
\begin{proof}
	By definition, $\omega(\ell_{i+1}) > \omega(\ell_i) + \omega(X[i+2])$.
	Further, $N(\ell_{i+1}) = N(\ell_{i}) \cup \{X[i+2]\}$.
	The lemma then follows.
\end{proof}

\subsection{Proof of \cref{thm:wis_hard}}

We reduce from the well-known \W[1]-hard \textsc{Multi-colored Independent Set} problem: Given a graph $G=(V_1 \cup \dots \cup V_k, E)$ where~$k \geq 3$ and $V_i$ is a clique for each~$i\in[k]$, is there a independent set of size $k$?
Without loss of generality, we assume $|V_k| = 1$.
Let $V:= V_1 \cup \dots \cup V_k$ and $n := |V|$.

By \cref{thm:wis_all_exp}, there exists an instance~$(U,F,\omega)$ of \textsc{Weighted Independent Set/3-Swap} with~$\Omega(n)$ vertices such that all local optima are exponentially far away from some initial solution~$S\subseteq U$ in the transition graph of~$(U,F,\omega)$; moreover, such an instance can be constructed in polynomial time via the provided reduction.
Assume without loss of generality that each weight of~$\omega$ is divisible by~$8n$ and let~$\omega_{\max}$ denote the largest assigned weight.

Now we construct an instance of \pivotprob{Weighted Independent Set}{3-Swap} consisting of a graph $G' := (V',E')$, weight function $\omega'$, and an initial solution $T$ as follows.
Initialize~$G'$ as the disjoint union of $G$ and $(U, F)$.
Each vertex in $V$ has weight one, and each vertex in $U$ carries the same weight as assigned by $\omega$.
Next, add a set~$X$ of~$k-2$ vertices~$x_i$ with~$i\in [2,k-1]$.
Each such vertex has weight~$3$.
For each~$i\in [2,k-1]$, make~$x_i$ adjacent to all vertices of~$V_{i-1} \cup V_{i} \cup V_{i+1}$.
Next, add an up-elevator $Y$ of $X$, where we consider the ordering $(x_2, \dots, x_{k-1})$ of the vertices in $X$.
We label the vertices in $Y$ as $y_3, \dots, y_{k-1}$ such that the neighborhood of $y_i$ in $X$ is exactly $\{x_i \mid j \in [2,i]\}$.
Moreover, we add two vertices~$v^*$ and~$w^*$ of weight~$2\omega_{\max}$ and~$3\omega_{\max} + 1$, respectively. Then, we add the edge~$\{v^*,w^*\}$  and make~$w^*$ is adjacent to all vertices of~$V \cup X \cup Y$.  
Finally, we make each vertex of~$V \cup X \cup Y \cup \{v^*\}$ adjacent to each vertex of~$U \setminus S$.
See \cref{fig:w1_hard} for an illustration.

\begin{figure}[ht!]
	\centering
	\includegraphics{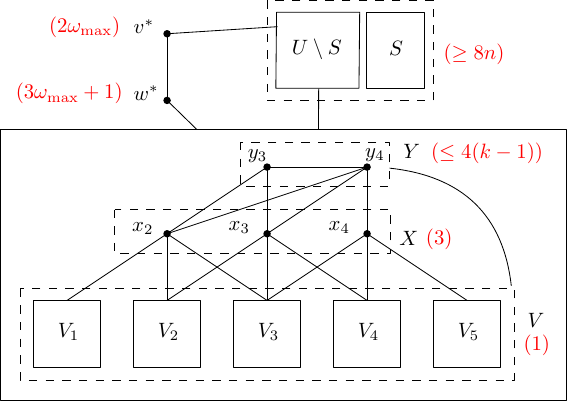}
	\caption{Illustration of $G'$ in the proof of \cref{thm:wis_hard} for $k = 5$. Not all edges are depicted. An edge between a box and a vertex/box indicate a complete bipartite subgraph between the vertices in the former box and the latter vertex/box. Red numbers in brackets indicate the weights.}
	\label{fig:w1_hard}
\end{figure}

The initial solution~$T$ is then defined by~$S \cup V_k \cup \{v^*\}$.
Recall that~$V_k$ contains only a single vertex and that this vertex is part of every multi-colored independent set of size~$k$ in~$G$.

\begin{claim}
\label{claim:wis_hard}
	Let $\pi$ be a maximal improving sequence for $(G', \omega')$ from $T$.
	If $\pi$ contains a solution $\overline{T}$ such that $\overline{T} \cap V$ is an independent set of size $k$, then $\overline{T}$ is locally optimal.
	Otherwise, $\pi$ has length exponential in $n$. 
\end{claim}
\begin{claimproof}
We start with the following observations:
\begin{enumerate}
\item\label{no w star} There is no improving swap that removes~$w^*$.
\item\label{w star in implies all exp} If~$w^*$ is part of a solution, then the solution is a subset of~$U \cup\{w^*\}$.
\item\label{w star or v star} Each reachable solution from~$T$ contains either~$v^*$ or~$w^*$.
\end{enumerate}
Indeed, Item~\ref{no w star} follows from the fact that~$w^*$ has weight~$3\wmax + 1$, which is strictly larger than the sum of the next two highest weights assigned by~$\omega'$.
Thus, no improving~$3$-swap can remove vertex~$w^*$ from a given solution.
Next, Item~\ref{w star in implies all exp} follows from the fact that~$w^*$ is a neighbor of each other vertex of~$V' \setminus U$.
Finally, Item~\ref{w star or v star} follows from Item~\ref{no w star} and the fact that each reachable solution from~$T$ that does not contain~$w^*$, has to contain~$v^*$, since~$v^*$ is contained in~$T$ and the weight of~$v^*$ is at least the weight of any two vertices besides~$\{v^*, w^*\}$.
In other words, each improving swap that removes~$v^*$ from the current solution, has to add~$w^*$ to the solution.

Now suppose $\pi$ has a solution $\overline{T}$ such that $\overline{T} \cap V = \{v_1, \dots, v_k\}$ is an independent set.
Without loss of generality, assume that $v_i \in V_i$ for $i \in [k]$.
For the sake of contradiction, assume that $\overline{T}$ is not locally optimal and there is an improving~$3$-swap~$W$.
Since the swap is improving, at least one vertex is added by this swap.
Note that this swap cannot add any vertex of~$Y \cup \{w^*\} \cup (U\setminus S)$ to the solution, since each such vertex is adjacent to~$v_1,v_2,$ and~$v_3$, and at most two of these vertices can be removed by~$W$.
Similarly, $W$ cannot add any vertex~$x_i$ to the solution, since~$x_i$ is adjacent to~$v_{i-1},v_i,$ and~$v_{i+1}$ and at most two of these vertices can be removed by~$W$.
Further, if $v^*$ is not in $\overline{T}$ then it cannot be added by~$W$, due to Items~\ref{no w star} and~\ref{w star or v star}.
Hence, the only vertices~$W$ could add are from~$V$.
However, if $W$ adds a vertex in $V_i$ for some $i \in [k-1]$, it needs to remove $v_{i}$ and we cannot add other vertices from $V_i$, since $V_i$ is a clique.
This implies that $W$ can only add at most one vertex in $V$ and remove at least one vertex in $V$.
Since, vertices in $V$ have the same weights and vertices in $V' \setminus V$ have positive weights and cannot be added by $W$, $W$ cannot be improving, a contradiction.
Hence, $\overline{T}$ is locally optimal.

For the remainder of the proof, we now assume that $\pi$ does not have such a solution~$\overline{T}$.
We show that if $\pi$ contains the solution~$S \cup \{w^*\}$, then it has length exponential in $n$.
Indeed, by Items~\ref{no w star} and~\ref{w star in implies all exp}, the subsequence~$\pi'$ of~$\pi$ starting from~$S\cup \{v^*\}$ is a sequence of improving~$3$-swaps for~$(U,F,\omega)$ starting from~$S$.
Since $S$ is far away from all the local optimum of~$(U,F,\omega)$, the sequence~$\pi'$ has exponential length in~$n$ as claimed.

Hence, it suffices to show that $\pi$ must contain the solution~$S \cup \{w^*\}$.
For each~$i\in [1,|\pi|]$, let~$T_i := \pi(i)$ denote the~$i$th solution in this sequence.
By Item~\ref{w star or v star}, for each~$i\in [1,|\pi|]$, $T_i$ contains either~$v^*$ and~$w^*$.
Moreover, by Item~\ref{no w star}, the solutions in~$\pi$ containing~$v^*$ are a prefix of~$\pi$.
Let~$i\in [1,|\pi|]$, such that~$v^*\in T_i$.
Observe that this implies that~$T_i \cap U = S$.
This is due to the following facts:
\begin{itemize}
\item $v^* \in T_i$, which implies that~$T_i \cap U \subseteq S$, since~$v^*$ is a neighbor of each vertex of~$U \setminus S$.
\item Since $Y$ is a clique, $T_i$ has at most one vertex in $Y$, whose weight is less than $4k$. Hence, the total weight of all vertices in~$T_i \cap (X \cup Y \cup V)$ is less than $8n$, and hence less than the weight of each individual vertex of~$U$.
\end{itemize}
That is, $T_i \cap U$ contains only vertices of~$S$, and cannot miss any vertex of~$S$, as otherwise, the weight of~$T_i$ would be strictly less than the weight of~$T$, which would then contradict the fact that~$T_i$ is part of an improving sequence starting with~$T$.

Note that this further implies that~$T_i$ contains at least one vertex of~$V \cup X \cup Y$, as otherwise, $T_i = S \cup \{v^*\}$, which has strictly less weight than~$T = S \cup \{v^*\} \cup V_k$.
Combined with Item~\ref{w star in implies all exp}, this then implies that if~$T_{i+1}$ exists and contains~$w^*$, the required~$3$-swap (i)~adds~$w^*$ to the solution and (ii)~removes~$v^*$ and a unique vertex of~$V \cup X \cup Y$ from~$T_i$.
Hence, if~$T_{i+1}$ exists and contains~$w^*$, then~$T_{i+1} \cap U = T_i \cap U = S$, which implies that~$T_{i+1} = S \cup \{w^*\}$.

Consequently, to finish the proof, it suffices to show that~$T_i$ is not locally optimal.
Since~$\pi$ is a sequence of finite length, this then implies that there is some smallest~$j$, such that~$T_j$ contains~$w^*$ and moreover fulfills~$T_j = S \cup \{w^*\}$ by the above argumentation.
To show that~$T_i$ is not locally optimal, we distinguish several cases depending on the intersection of~$T_i$ with~$V \cup X  \cup Y$.
For each of the cases, we present an improving~$3$-swap.
Let~$T_i' := T_i \cap (V \cup X  \cup Y)$ and note that~$T_i = S \cup \{v^*\} \cup T_i'$.

\textbf{Case 1:} \emph{$T_i' \subseteq V$.}
Recall that for each~$j\in [k]$, $V_j$ is a clique in~$G'$, which implies that~$|T_i' \cap V_j| \leq 1$. 
Since $T_i \cap V$ is not an independent set of size~$k$, there is some~$p\in [2,k-1]$, such that~$|T_i' \cap (V_{p-1} \cup V_{p} \cup V_{p+1})| \leq 2$.
Hence, the swap adding the vertex~$x_p$ to the solution while removing the at most two vertices from~$V_{p-1} \cup V_{p} \cup V_{p+1}$ is an improving~$3$-swap.

\textbf{Case 2:} \emph{$T_i' \subseteq V \cup X$ with~$T_i' \cap X \neq \emptyset$.}
If~$T_i' \cap X = X$, then~$T_i' = X$ and (i)~adding~$y_3$ to the solution and (ii)~removing~$x_2$ and~$x_3$ from the solution, yields an improving~$3$-swap.
Otherwise, there is some~$p\in [2,k-2]$, such that either~$x_p \in T_i'$ and~$x_{p+1} \notin T_i'$ or vice versa.
In both cases, $T_i' \cap (V_{p} \cup V_{p+1}) = \emptyset$, since~$x_p$ and~$x_{p+1}$ are both adjacent to each vertex of~$V_p$ and $V_{p+1}$.
Assume without loss of generality that~$x_p \in T_i'$ and~$x_{p+1} \notin T_i'$.
Then, (i)~adding~$x_{p+1}$ to the solution and (ii)~removing at most one vertex of~$V_{p+2}$ from the solution yields an improving~$3$-swap.

\textbf{Case 3:} \emph{There is some~$p\in [3,k-1]$, such that~$y_p \in T_i'$.}
If $p < k-1$, then there exists an improving 3-swap by \cref{lem:elevator_impr}.
Otherwise, observe that $y_{k-1}$ is the neighbor of all vertices in $V \cup X \cup Y$.
Hence, $T_i \cap (V \cup X \cup Y) = \{y_{k-1}\}$.
Then adding $w^*$ and removing $v^*$ and $y_{k-1}$ yields an improving 3-swap.

Hence, if~$G$ does not contain an independent set of size~$k$, then $S\cup \{w^*\}$ is a solution which is part of every sequence of improving~$3$-swaps starting at~$T$.
This completes the proof of the claim. 
\end{claimproof}

Note that the claim above implies that if $G$ has an independent set $Q$ of size $k$, then there is a maximal improving sequence starting from $T$ with length at most $k$.
Indeed, we add all the vertices of $Q$ to the initial solution by~$k-1$ consecutive~$1$-swaps.
Then we arrive at a solution $\overline{T}$ such that $\overline{T} \cap V = Q$.
Hence, $\overline{T}$ is locally optimal by \cref{claim:wis_hard}.

Now suppose for the sake of contradiction that there exists an algorithm to solve \pivotprob{Weighted Independent Set}{3-Swap} in time $g(\ell) \cdot n^{\Oh(1)}$ for some computable function~$g$.
We set $\ell = k$, and simulate the algorithm on the instance $(G', \omega', T)$ above for $g(\ell) \cdot n^{\Oh(1)}$ steps.
If the algorithm returns a solution $T^*$ such that $T^* \cap V$ is an independent set of size $k$, then we know that $G$ is a YES-instance of \textsc{Multi-colored Independent Set}.
Otherwise, by \cref{claim:wis_hard}, since the running time $g(\ell) \cdot n^{\Oh(1)}$ is not exponential in $n$, the algorithm should not return any locally optimal solution; that is, it either returns a solution that is not locally optimal or does not terminate within this time.
In this case, we know that the promise that there is a local optimum within $k$ improving steps from $T$ is broken.
As argued above, this implies that $G$ does not have an independent set of size $k$, and hence~$(G,k)$ is a NO-instance.
Therefore, we can decide \textsc{Multi-colored Independent Set} in time $g(k) \cdot n^{\Oh(1)}$, a contradiction to the assumption $\FPT \neq \W[1]$. \hfill $\qedsymbol$

\begin{remark}
	Our proof of \cref{thm:wis_hard} above also implies the following:
	\begin{itemize}
		\item It is \W[1]-hard to decide whether there is a local optimum of distance at most $k$ from a given initial solution, when parameterized by $k$.
		\item Unless $\FPT = \W[1]$, there is no algorithm to approximate the distance to the nearest local optimum with running time $g(\ell) \cdot n^{\Oh(1)}$ and approximation ratio $n^{\Oh(1)}$.
	\end{itemize}
\end{remark}

\subsection{All-exp Property of \plsprob{Weighted Independent Set}{3-swap}}
\label{subsec:wis_all_exp}

The main argument to show the all-exp property of \plsprob{Weighted Independent Set}{3-swap} is the following tight \pls-reduction from \plsprob{Max Cut}{Flip} (\cref{thm:tight_to_wis} below).
This reduction is a slight adaptation to the \pls-reduction presented by Komusiewicz~and~Morawietz~\cite{KomusiewiczM25}. 
Roughly speaking, the reduction constructs an instance with two parts~$C$ and~$D$ where the vertices of $C$ model the respective solutions of \textsc{Max Cut} and the vertices of~$D$ are used as gadgets to ensure that we can simulate the flip of a vertex in the \textsc{Max Cut} instance by a sequence of improving 3-swaps.
In the original reduction by Komusiewicz~and~Morawietz, some maximal independent sets can be improved by only swapping vertices in~$C$, thus leading to intermediate independent sets that are not maximal. Such independent sets could be improved by 1-swaps which could be combined with improving 1-swaps or 2-swaps in distant parts of the graph. Consequently, there is not a strong correspondence between improving flips in the \textsc{Max Cut} instance and improving swaps in the \textsc{Independent Set} instance.
To establish this correspondence, we add further vertices to~$C$, thus ensuring that each maximal solution can only be improved by swapping at least one vertex of~$D$ into the solution.
By turning~$D$ into a clique, we now ensure that at each time step the flip of only a single vertex in the \textsc{Max Cut} instance can be simulated.

Although \plsprob{Max Cut}{Flip} is already defined in \cref{sec:framework}, we give a slightly different definition here that is easier to work with.
An instance $I$ of \plsprob{Max Cut}{Flip} consists of an undirected graph $G = (V, E)$ and a weight function $\omega : E(G) \to \mathbb{N}$.
A solution of $I$ (also called a \emph{cut}) is a partition $(A, B)$ of the vertex set $V$ (i.e., each vertex of $V$ belongs to exactly one of $A$ and $B$), and its objective value is $\sum_{(u,v)\in A\times B, uv\in E} \omega(uv)$.
Two solutions are neighbors, if they differ by the \emph{flip} of a vertex, defined as the act of moving the vertex from one set of a partition to the other.

\begin{theorem}

\label{thm:tight_to_wis}
There exists a tight \pls-reduction from \plsprob{Max Cut}{Flip} restricted to instances $(G, \omega)$ of maximum degree at most $(\log|V(G)|)^{\Oh(1)}$ to \plsprob{Weighted Independent Set}{3-swap}.
\end{theorem}

\begin{proof}
\textbf{The core graph.}
Let~$I:=(G=(V,E),\omega)$ be an instance of \plsprob{Max Cut}{Flip} with maximum degree $\Delta$, where $\Delta = (\log|V|)^{\Oh(1)}$.
Let $n:= |V|$ and assume that each weight assigned by~$\omega$ is a multiple of~$2n$.
We obtain an instance~$I':=(G':=(V',E'),\omega')$ of \plsprob{Weighted Independent Set}{3-swap} as follows:
For each vertex~$v\in V$, we add six vertices~$v_A,v_A',v_A'',v_B,v_B'$, and~$v_B''$.
The weight of~$v_A$ and~$v_B$ is~$\alpha \cdot 4n$ and the weight of the other four vertices is~$\alpha := 2 \cdot \omega(E)$. For each edge~$uv\in E$, we add two vertices~$x_{u,v}$ and~$x_{v,u}$, each of weight~$\omega(uv)$.
Between these vertices, we add the following edges:
For each vertex~$v\in V$, we make~$G'[\{v_A,v_A',v_A'',v_B,v_B',v_B''\}]$ a biclique with bipartition~$(\{v_A,v_A',v_A''\},\{v_B,v_B',v_B''\})$.
Moreover, for each neighbor~$u$ of~$v$ in~$G$, we add (i)~edges between~$x_{u,v}$ and each vertex of~$\{v_A,v_A',v_A''\} \cup \{x_{v,w}\mid w \in N_G(v)\}$ and (ii)~edges between~$x_{v,u}$ and each vertex of~$\{v_B,v_B',v_B''\} \cup \{x_{w,v}\mid w \in N_G(v)\}$. 
This completes the basic part of the graph~$G'$, we let~$C$ denote the vertices of this part of~$G'$. 
To complete the construction, we will add several up-elevators and down-elevators that come in pairs (together with one additional vertex each) to simulate the flip of a single vertex in the Max Cut instance by a sequence of improving~$3$-swaps.
Before we do this, we give an intuition of the inner workings of basic part of~$G'$.

\textbf{Intuition.}
The idea behind~$G'[C]$ is that there is a natural 1-to-1 correspondence between partitions~$(A,B)$ of~$G$ and maximal independent sets in~$G'[C]$:
For each vertex~$v\in V$, each maximal independent set~$S$ contains either all of~$\{v_A,v_A',v_A''\}$ or all of~$\{v_B,v_B',v_B''\}$.
This then directly gives a partition of the vertex set of~$V$ based on whether~$S$ contains all of~$\{v_A,v_A',v_A''\}$ or all of~$\{v_B,v_B',v_B''\}$.
Additionally, for the respective cut~$(A,B)$, for each edge~$uv$ with~$u\in A$ and~$v\in B$, the vertex~$x_{u,v}$ will also be part of~$S$.
Thus, $\omega'(S) = n\cdot (4n+2)\alpha + \sum_{(p,q)\in A\times B, pq\in E} \omega(pq)$.
To simulate a flip in the partition~$(A,B)$, there thus exists a maximal independent set~$S'$ corresponding to the new partition~$(A',B')$ such that the symmetric difference between~$S$ and~$S'$ has size~$\Oh(\Delta)$.
As we aim for a hardness reduction with only~$3$-swaps, we however need to add further gadgets to ensure that we can simulate such an improving $\Oh(\Delta)$-swap by a sequence of improving~$3$-swaps.
In the original reduction by Komusiewicz~and~Morawietz~\cite{KomusiewiczM25}, the vertices~$\{v_A',v_A'',v_B',v_B''\mid v\in V\}$ did not exist, which allowed for swaps that removed a vertex~$v_A$, added a vertex~$v_B$, and additionally added some vertex~$x_{w,v}$.
In our reduction, this is not possible anymore, since the vertices~$v_A'$ and~$v_A''$ prevent the addition of~$v_B$ if at least one of these vertices is contained in the current solution as well.
Essentially, this ensures that if we have a maximal independent set~$S$ that contains for each vertex~$v$ either all of~$\{v_A,v_A',v_A''\}$ or all of~$\{v_B,v_B',v_B''\}$, then we can only find an improving swap by adding a vertex outside of~$C$ to the solution.

\begin{figure}[ht!]
	\centering
	\includegraphics[page=1]{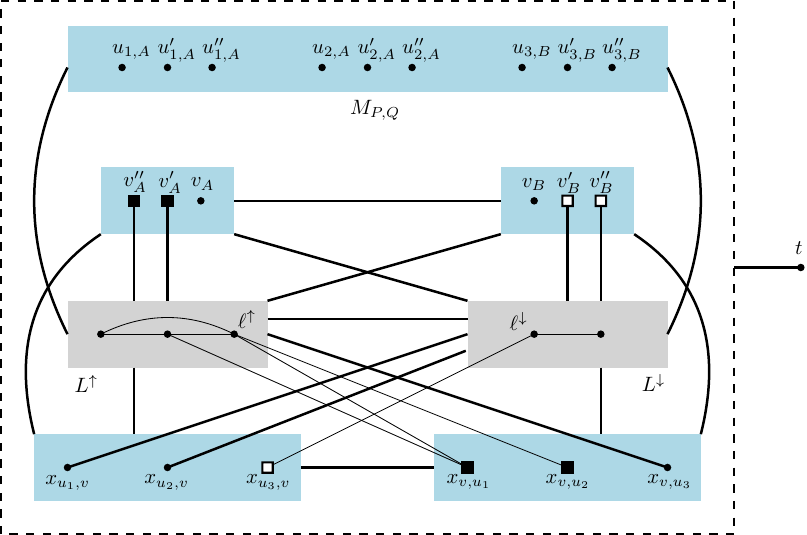}
	\caption{Illustration of the neighborhood of an $A2B$-simulator $(\Lup, t, \Ldo)$ of a vertex $v$ for $(P,Q$), where $P = \{u_1, u_2\}$ and $Q = \{u_3\}$. Vertices in $\Xup_{P, Q}$ and $\Xdo_{P,Q}$ are represented by black and white squares, respectively. Bold edges represent bicliques between a box and another box/vertex. Vertices in blue boxes are in $C$.}
	\label{fig:a2b_sim}
\end{figure}

\textbf{Additional Elevators.}
We now describe the remaining parts of~$G'$ that are outside of~$C$.
For each vertex~$v\in V$, we compute the collection~$\mr_v$ of all partitions~$(P,Q)$ of the neighborhood of~$v$ in~$G$ for which~$\sum_{p\in P} \omega(vp) < \sum_{q\in Q} \omega(vq)$.
Note that this can be done in polynomial time, since the maximum degree $\Delta$ of the graph is~$(\log n)^{\Oh(1)}$.
The collection~$\mr_v$ contains exactly those partitions~$(P,Q)$ for which flipping vertex~$v$ is improving, if~$Q$ are exactly those neighbors of~$v$ that are currently on the same side of the cut as~$v$.

Let~$v$ be a vertex of~$V$ and let~$(P,Q)$ be a partition of~$\mr_v$.
We first describe the elevator gadgets to simulate the flip of vertex~$v$ from~$A$  to~$B$; see \cref{fig:a2b_sim} for an illustration.
The gadgets for the flip from~$B$ to~$A$ are then nearly identical.
To formally define elevator gadgets, we need to define an ordering of the vertices of the graph.
We define this to be an arbitrary ordering of the vertices where the vertices of~$\{v_A',v_A'',v_B',v_B''\mid v\in V\}$ precede all other vertices.

We let~$\Xup_{P,Q} := \{v_A',v_A''\}\cup \{x_{v,p} \mid p\in P\}$, $\Xdo_{P,Q} := \{v_B',v_B''\} \cup \{x_{q,v} \mid q\in Q\}$, $M_{P,Q}:= \{p_A, p_A', p_A'' \mid p\in P\} \cup \{q_B, q_B', q_B'' \mid q\in Q\}$.
We add an~$\Xup_{P,Q}$-up-elevator~$\Lup$ to~$G'$, where the neighborhood of the vertices of~$\Lup$ outside of~$\Xup_{P,Q} \cup \Lup$ is~$M_{P,Q} \cup \{v_B,v_B',v_B''\} \cup \{x_{v,q} \mid q\in Q\} \cup \{x_{r,v} \mid r\in N_G(v)\}$.
Additionally, we add an~$\Xdo_{P,Q}$-down-elevator~$L^\downarrow$ to~$G'$, where the neighborhood of the vertices of~$\Ldo$ outside of~$\Xdo_{P,Q} \cup \Ldo$ is~$M_{P,Q} \cup \{v_A,v_A',v_A''\} \cup \{x_{p,v} \mid v\in P\} \cup \{x_{v,r} \mid r\in N_G(v)\}$.
Note that the ordering of the vertices implies that the bottom-level of~$\Lup$ is adjacent to only~$v_A'$ and~$v_A''$ of~$\Xup_{P,Q}$ and the bottom-level of~$\Ldo$ is adjacent to only~$v_B'$ and~$v_B''$ of~$\Xdo_{P,Q}$.
In particular, the weight of each vertex in any of these elevators is at least~$2\alpha-1$ (which is the weight of the bottom-level of~$\Ldo$) and at most~$3\alpha + n < 4\alpha$.

We make~$\Lup \cup \Ldo$ a clique.
Note that the closed neighborhood of the top-level~$\lup$ of~$\Lup$ contains exactly $v_B$ and the vertices of~$N_{P,Q} := M_{P,Q} \cup \{v_A',v_A'',v_B',v_B''\} \cup \{x_{v,r}, x_{r,v} \mid r\in N_G(v)\} \cup L^\uparrow \cup L^\downarrow$.
Similarly, the closed neighborhood of the top-level~$\ell^\downarrow$ of~$L^\downarrow$ is $N_{P,Q} \cup \{v_A\}$.
Next, we add a~\emph{turn vertex}~$t$ with weight equal to~$\omega'(\ell^\downarrow) + \omega'(v_A) - 1$, which we make adjacent to all vertices of~$N_{P,Q} \cup \{v_A,v_B\}$.  
We call the tuple~$(L^\uparrow,t,L^\downarrow)$ an~\emph{$A2B$-simulator of~$v$ for partition~$(P,Q)$}.

Note that since~$\Lup$ is an up-elevator with $|P|+1$ levels, $\omega'(\ell^\uparrow) = \omega'(\Xup_{P,Q}) + |P|+1 = 2\alpha + |P| + 1 + \sum_{p\in P} \omega(vp)$.
Further, since~$\Ldo$ is a down-elevator with $|Q|+1$ levels, $\omega'(\ell^\downarrow) = \omega'(\Xdo_{P,Q}) - |Q|-1 = 2\alpha - |Q| - 1 + \sum_{q\in Q} \omega(vq)$.
By the fact that $(P,Q) \in \mr_v$  and that each edge weight is a multiple of~$2n$, $\sum_{q\in Q} \omega(vq) - \sum_{p\in P} \omega(vp) \geq 2n$. 
Since~$|P| + |Q| < n$ and~$n>2$, this implies that
\begin{equation}
\label{eq:top_level_weights}
	\omega'(\ell^\uparrow) + 1 < \omega'(\ell^\downarrow).
\end{equation}

We add symmetrically a~\emph{$B2A$-simulator of~$v$ for~$(P,Q)$}, for which we exchange in the definition the subscripts of all relevant vertices from~$A$ to~$B$ and vice versa, and where we exchange vertices~$x_{r,s}$ by~$x_{s,r}$.

Finally, we let~$D$ denote the set of all vertices of~$G'$ that are not in~$C$ and turn these vertices into a clique.
This ensures that each independent set of~$G'$ can contain at most one vertex that is not in~$C$.
This is the second difference to the reduction by Komusiewicz~and~Morawietz~\cite{KomusiewiczM25}, where in~$D$, each maximal clique only had constant size.
However, to ensure the tightness of the reduction, we need to make the whole set~$D$ into a clique.

\textbf{Correctness.}
We now show that this reduction has the stated properties.
Let~$(A,B)$ be a partition of~$V$.
We define
$$g(A,B) := \{v_A,v_A',v_A''\mid v\in A\} \cup \{w_B,w_B',w_B''\mid w\in B\} \cup \{x_{v,w}\mid v\in A, w\in B, \{v,w\}\in E\}.$$

\begin{claim}
For each partition~$(A,B)$ of~$V$, $g(A,B)$ is an independent set in~$G'$.
\end{claim} 
\begin{claimproof}
Note that~$g(A,B)$ contains only vertices of~$C$, that is, no vertices of elevators and no turn vertices.
By construction, the vertices~$\{v_A,v_A',v_A''\mid v\in A\} \cup \{w_B,w_B',w_B''\mid w\in B\}$ form an independent set.
It remains to show that each vertex~$x_{v,w}\in g(A,B)$, $x_{v,w}$ has no neighbor in~$g(A,B)$.
The containment of~$x_{v,w}$ in~$g(A,B)$ implies that~$v\in A$ and~$w\in B$.
Consider the neighbors of~$x_{v,w}$ in~$C$.
These are (i)~$\{w_A,w_A',w_A''\}$, (ii)~$\{v_B,v_B',v_B''\}$, and (iii)~vertices of the form~$x_{s,t}$ with~$s=w$ or~$t=v$.
Since~$(A,B)$ is a partition of~$V$ and~$x_{v,w}\in g(A,B)$, $g(A,B)$ contains none of the vertices of~$\{w_A,w_A',w_A'',v_B,v_B',v_B''\}$. 
Moreover, $g(A,B)$ does not contain a vertex~$x_{s,t}$ with~$s=w$ or~$t=v$, since this would imply that~$s\in A$ and~$t\in B$ which contradicts~$v\in A$ and~$w\in B$ and the fact that~$(A,B)$ is a partition of~$V$.
Consequently, $x_{v,w}$ has no neighbor in~$g(A,B)$, which implies that~$g(A,B)$ is an independent set.
\end{claimproof}

We have the following claim that will be useful later.

\begin{claim}\label{C no loc opt}
Let $S$ be an independent set of $G'$ such that $S$ contains a vertex $z \in D$.
Without loss of generality, assume that $z$ is in an~$A2B$-simulator~$(\Lup,t,\Ldo)$ of some vertex~$v\in V$ for some partition~$(P,Q) \in \mr_v$.
Let $\Lup = \{\lup_1, \dots, \lup_{|P|+1}\}$ and $\Ldo = \{\ldo_1, \dots, \ldo_{|Q|+1}\}$.
Then there exists an independent set $S'$ such that $|S \oplus S'| \leq 3$ and $\omega'(S') > \omega'(S)$.
In particular,
\begin{itemize}
	\item If $z = \lup_i$ for some $i \in [1,|P|]$, then $S' = (S \setminus \{z, \Xup_{P,Q}[i+2]\}) \cup \{\lup_{i+1}\}$.
	\item If $z = \lup_{|P|+1}$, then $S' = (S \setminus \{z, v_A\}) \cup \{t\}$.
	\item If $z = t$, then $S' = (S \setminus \{z\}) \cup \{\ldo_{|Q|+1}, v_B\}$.
	\item If $z = \ldo_i$ for some $i \in [2, |Q|+1]$, then $S'$ is either  $(S \setminus \{z\}) \cup \{\ldo_{i-1}, \Xdo_{P,Q}[i+1]\}$ when this is an independent set or $(S \setminus \{z\}) \cup \{q_B\}$ for $q$ such that $x_{q,v} = \Xdo_{P,Q}[i+1]$.
	\item If $z = \ldo_1$, then $S' = (S \setminus \{z\}) \cup \{v_B', v_B''\}$.
\end{itemize}
\end{claim}
\begin{claimproof}
We consider the five cases in the statement.

\textbf{Case 1:} \textit{$z = \lup_i$ for some $i \in [1,|P|]$.}
Then the statement follows from \cref{lem:elevator_impr}.

\textbf{Case 2:} \textit{$z = \lup_{|P|+1}$.}
By \eqref{eq:top_level_weights} and by construction, observe that $\omega'(z) + \omega'(v_A) < \omega'(\ldo_{|Q|+1}) + \omega'(v_A)  - 1 = \omega'(t)$.
Moreover, $N[t] = N[z] \cup \{v_A\}$.
Thus, $S' := (S\setminus \{z,v_A\}) \cup \{t\}$ satisfies the condition of the claim.

\textbf{Case 3:} \textit{$z = t$.}
By definition, $N[t] = \{v_B\} \cup N[\ldo_{|Q|+1}]$ and~$\omega'(t) = \omega'(\ldo_{|Q|+1}) + \omega'(v_B) - 1$.
Hence, $S' := (S\setminus \{t\}) \cup \{\ldo_{|Q|+1},v_B\}$ is a strictly better independent set than~$S$, and $|S \oplus S'| \leq 3$.

\textbf{Case 4:} \textit{$z = \ldo_i$ for some $i \in [2, |Q|+1]$.}
By definition, $N[\ldo_i] = N[\ldo_{i-1}] \cup \{x_{q,v}\}$, where $x_{q,v} := \Xdo_{P,Q}[i+1]$.
Moreover, $\omega'(\ldo_{i-1}) + \omega'(x_{q,v}) = \omega'(\ldo_i) + 1$.
Hence, $S' := (S\setminus \{z\}) \cup \{\ldo_{i-1}, x_{q,v}\}$ has strictly higher weight than~$S$ and $|S \oplus S'| \leq 3$.
If~$S'$ is an independent set, we are done.
Otherwise, $S$ contains a neighbor~$z'$ of~$x_{q,v}$.
Since~$x_{q,v}$ is not adjacent to any of the vertices in~$\{\ldo_{i-1}, v_B, v_B',v_B'',q_A,q_A',q_A''\} \cup \{x_{w,v}\mid w\in N_G(v)\} \cup \{x_{q,u}\mid u\in N_G(q)\}$, $z'$ is none of these vertices.
Moreover, $\{q_B\} \cup \{x_{v,w}\mid w\in N_G(v)\}$ do not intersect with $S$, since~$z$ is adjacent to all these vertices and~$S$ is an independent set containing~$z$.
Thus, $z'$ has to be one of the remaining neighbors of~$x_{q,v}$ inside~$C$, which are the vertices of~$\{q_B',q_B''\}\cup\{x_{u,q}\mid u\in N_G(q)\}$.
By construction, this implies that~$z'$ is adjacent to all neighbors of~$q_B$ outside of~$C$.
Hence, $(S\setminus \{z\}) \cup \{q_B\}$ is an independent set of higher weight than~$S$, since~$\omega'(q_B) = 4\alpha n > \omega'(\ldo_i)$.

\textbf{Case 5:} \textit{$z = \ldo_1$.}
By definition, $\omega'(v_B') + \omega'(v_B'') = \omega'(z) + 1$ and~$N[v_B'] \cup N[v_B''] \subseteq  N[z]$.
Hence, $S':= (S\setminus \{z\}) \cup \{v_B',v_B''\}$ is a strictly better independent set than~$S$.
\end{claimproof}

One immediate implication of the claim above is that all locally optimal independent sets of~$G'$ contain only vertices of~$C$.
Next, let~$R:= \{g(A,V\setminus B)\mid A\subseteq V\}$.

\begin{claim}
\label{claim:R_all_lopt}
The set~$R$ contains all locally optimal independent sets of~$G'$.
\end{claim} 
\begin{claimproof}
Let $S$ be a locally optimal independent set of $G'$.
By~\Cref{C no loc opt}, $S$ contains only vertices of~$C$.
First, we show that $S$ contains exactly one of $v_A$ and $v_B$.
If~$v_A$ and~$v_B$ are both not contained in~$S$, then~$S \cup \{v_A\}$ or~$S \cup \{v_B\}$ is an independent set.
This is due to the fact that~$N[\{v_A,v_B\}] \cap C$ is a biclique.
Hence, the two adjacent vertices $v_A$ and~$v_B$ cannot both have neighbors in~$S\cap C$.
Further, since~$S$ contains no vertex outside of~$C$, $S \cup \{v_A\}$ or~$S \cup \{v_B\}$ is an independent set of strictly larger weight than~$S$, a contradiction.
Since $v_A$ and $v_B$ are adjacent, exactly one of them is contained in $S$.

Since $v_A$, $v'_A$, and $v_A''$ are not adjacent and have the same neighbors in $C$ as $v_A$, if $v_A \in S$ then $v'_A, v''_A \in S$.
With a similar argument for $v'_B$ and $v_B''$, $S$ must contain for each vertex~$v\in V$ exactly one of the two sets~$\{v_A,v_A',v_A''\}$ and~$\{v_B,v_B',v_B''\}$.
Thus, for each locally optimal independent set~$S$, there is a partition~$(A,B)$ of~$V$, such that~$C \supseteq S\supseteq \{v_A,v_A',v_A''\mid v\in A\} \cup \{w_B,w_B',w_B''\mid w\in B\}$.
For each edge~$uv\in E$ with~$u\in A$ and~$v\in B$, the vertex~$x_{u,v}$ is also part of~$S$, since the only neighbors of~$x_{u,v}$ are vertices that are adjacent to at least one of~$u_A$ or~$v_B$ and thus not contained in~$S$.
In other words, if~$S$ would not contain~$x_{u,v}$, then~$S\cup \{x_{u,v}\}$ would be a better independent set.
Consequently, for each locally optimal independent set~$S$, there is a partition~$(A,B)$ of~$V$, such that~$S = g(A,B)$.
Thus, $S\in R$. 
\end{claimproof}

\begin{claim}
\label{claim:lopt_to_lopt}
For each partition~$(A,B)$ of~$V$, where~$g(A,B)$ is a locally optimal independent set in~$G'$, $(A,B)$ is a locally optimal cut for~$G$.
\end{claim} 
\begin{claimproof}
Let~$S$ be a locally optimal independent set in~$G'$ and let~$(A,B)$ be the partition of~$V$, such that~$S=g(A,B)$.
We show that~$(A,B)$ is a locally optimal cut for~$G$.
Assume towards a contradiction that this is not the case and let~$v$ be a vertex of~$V$ that (without loss of generality) can be flipped from~$A$ to~$B$ to obtain a better cut.
We show that there is an improving~$3$-swap for~$S$ in~$G'$.
Since flipping vertex~$v$ from~$A$ to~$B$ is improving, $\sum_{p\in P}\omega(\{v,p\}) < \sum_{q\in Q}\omega(\{v,q\})$ for~$P:=N_G(v) \cap A$ and~$Q:=N_G(v) \cap B$.
Thus, we added an~$A2B$-simulator~$(\Lup,t,\Ldo)$ for~$v$ with partition~$(P,Q)$.
Moreover, by definition of~$g(A,B)$, the bottom-level~$\ell_1$ of~$\Lup$ is adjacent to only the vertices of~$\{v_A',v_A''\}$ in~$S$ and has weight~$\omega'(\ell_1) = \omega'(v_A') + \omega'(v_A'') + 1$.
Hence, $S' := (S\setminus \{v_A',v_A''\}) \cup \{\ell_1\}$ is a strictly better independent set than~$S$.
This contradicts the assumption that~$S$ is locally optimal.
Consequently, $(A,B)$ is locally optimal if and only if~$g(A,B)$ is locally optimal.  
\end{claimproof}

For two independent sets $S$ and $S'$ in $R$, we define a \emph{direct sequence} from $S$ to $S'$ to be an improving sequence that starts at $S$, ends at $S'$, and contains no other solutions in $R$.

\begin{claim}
\label{claim:direct_sequence}
Let $(A,B)$ be a partition of $V$.
Let $S = g(A,B)$.
Then, 
\begin{enumerate}[(a)]
	\item For every partition $(\overline{A},\overline{B})$ of $V$ that can be obtained from $(A,B)$ by an improving flip, there exists a direct sequence from $S$ to $g(\overline{A},\overline{B})$.
	\item For any direct sequence $\pi$ from $S$ to another independent set $S'$ in $R$, there exists a partition $(\overline{A},\overline{B})$ that can be obtained from $(A,B)$ by an improving flip and satisfies $S' = g(\overline{A},\overline{B})$. Further, $\pi$ has length $\Oh(\Delta)$.
\end{enumerate}
\end{claim} 
\begin{claimproof}
We start with (a).
Without loss of generality, assume that $(\overline{A}, \overline{B})$ is obtained from $(A,B)$ by flipping a vertex $v$ from $A$ to $B$.
Let $P := N_G(v) \cap B$ and $Q := N_G(v) \cap A$.
Observe that $(P,Q) \in \mr(v)$.
Let $(\Lup, t, \Ldo)$ be the $A2B$-simulator of~$v$ for~$(P,Q)$.
Let $\Lup = \{\lup_1, \dots, \lup_{|P|+1}\}$ and $\Ldo = \{\ldo_1, \dots, \ldo_{|Q|+1}\}$, where $\lup_1$ and $\ldo_1$ are the bottom-levels of $\Lup$ and $\Ldo$, respectively.
Let $L$ be the sequence $\lup_1, \dots, \lup_{|P|+1}, t, \ldo_{|Q|+1}, \dots, \ldo_1$.
For $i \in [1,|P|+|Q|+3]$, let $S_i$ be the set obtained from $S$ by adding the $i$th element of $L$ and:
\begin{itemize}
	\item Removing the first $i+1$ elements of $\Xup_{P,Q}$, if $i \leq |P|+1$;
	\item Removing $\Xup_{P,Q}$ and $v_A$, if $i = |P|+2$;
	\item Removing $\Xup_{P,Q}$ and $v_A$, and adding $v_B$ and the last $i - |P|-3$ elements of $\Xdo_{P,Q}$, if $i > |P|+ 2$.
\end{itemize}
Further, define $S_{|P|+|Q|+4}$ to be the set obtained from $S$ by removing $\Xup_{P,Q}$ and $v_A$ and adding $v_B$ and $\Xdo_{P,Q}$.

We show that $\pi := (S, S_1, \dots, S_{|P|+|Q|+4})$ is an improving sequence.
By \cref{C no loc opt}, for $i \in [1,|P|+2] \cup \{|P|+|Q|+4\}$, $S_i$ is an independent set with larger weight than $S_{i-1}$.
Next, for $i \in [|P|+3, |P|+|Q|+3]$, we prove by induction that $S_i$ is an independent set.
For $i = |P|+3$, $S_i = (S_{i-1} \setminus \{t\}) \cup \{\ldo_{|Q|+1}, v_B\}$.
Note that $N[\ldo_{|Q|+1}] \cup (N[v_B] \cap C) \subseteq N[t]$.
Combined with the fact that $S_{i-1}$ is an independent set, this implies that $S_i$ is an independent set.
Now for $i \in [|P|+4, |P|+|Q|+3]$, let $i' = i - |P| - 4$.
we have $S_i = (S_{i-1} \setminus \{\ldo_{|Q|+1-i'}\}) \cup \{\ldo_{|Q|-i'}, x_{q,v}\}$ for some $q \in Q$.
Observe that $N[\ldo_{|Q|-i'}] \subseteq N[\ldo_{|Q|+1-i'}]$.
Further, a vertex in $N[x_{q,v}] \setminus N[\ldo_{|Q|-i'}]$ is either one of $q_A, q'_A, q''_A$ or in $D$.
Since $q \in Q \subseteq B$, $q_A, q'_A, q''_A$ are not in $S$ and hence not in $S_{i-1}$.
Further, an independent set can only contains at most one vertex in $D$, and hence, $S_{i-1}$ contains no vertex in $(N[x_{q,v}] \setminus N[\ldo_{|Q|-i'}]) \cap D$.
All of the above imply that $S_i$ is an independent set, completing the inductive proof.
Then by \cref{C no loc opt}, we also obtain that  for $i \in [|P|+3, |P|+|Q|+3]$, $S_i$ is an independent set with larger weight than $S_{i-1}$. 
Therefore, overall, $\pi$ is an improving sequence.

It is easy to see that $S_{|P|+|Q|+4}$ is indeed $g(\overline{A},\overline{B})$, and for $i \in [1,|P|+|Q|+3]$, $S_i \notin R$.
Hence, $\pi$ is a improving sequence from $S$ to $g(\overline{A},\overline{B})$ as required.

We now prove (b).
Let $\pi' = (S, S'_1, S'_2, \dots, S'_r)$ be a direct sequence.
Note that every vertex in $C \setminus S$ is adjacent to all vertices in either $\{v_A, v_A', v_A''\}$ or $\{v_B, v_B', v_B''\}$ for some vertex $v$ of $G$.
Hence, if $S'_1 \setminus S$ contains a vertex in $C$, then we need to remove at least three vertices from $S$; in this case, $S'_1$ cannot differ from $S$ by a 3-swap.
Therefore, $S'_1 \setminus S$ contains only vertices in $D$.
Since an independent set of $G'$ can only contain at most one vertex in $D$, this implies that $S'_1 \setminus S$ contains exactly one vertex $z \in D$.
Without loss of generality, suppose $z$ is in the $A2B$-simulator $(\Lup, t, \Ldo)$ of some vertex~$v$ for some partition~$(P,Q) \in \mr(v)$.
Note that $z$ is adjacent to $v_A', v_A'', v_B', v_B''$, and exactly two of these vertices are in $S$.
Hence, from $S$ to $S'_1$, we remove exactly $v_{\gamma}', v_{\gamma}''$ for some $\gamma \in \{A, B\}$ and add $z$.
In order for $S'_1$ to be an independent set, all neighbors of $z$ in $G$ cannot be in $S$.
These neighbors include all vertices in the set $M_{P,Q}= \{p_A, p_A', p_A'' \mid p\in P\} \cup \{q_B, q_B', q_B'' \mid q\in Q\}$.
In other words, $P \subseteq B$ and~$Q \subseteq A$.
Hence, if $v \in B$ (i.e., $\gamma = B$), then for every $q \in Q$, $x_{q,v} \in S$.
Since we cannot remove these $x_{q,v}$, $z$ cannot be adjacent to them.
This means that $z$ must be the bottom-level of $\Ldo$.
However, the swap that removes $v_{B}', v_{B}''$ and adds $z$ is then not improving.
Therefore, $v \in A$.

Let $\Lup = \{\lup_1, \dots, \lup_{|P|+1}\}$ and $\Ldo = \{\ldo_1, \dots, \ldo_{|Q|+1}\}$, where $\lup_1$ and $\ldo_1$ are the bottom-levels of $\Lup$ and $\Ldo$, respectively.
We define the sequence $L$ and the set $S_i$ for $i \in [1,|P|+|Q|+4]$ as in the proof of part (a).
We show that the sequence $\pi = (S, S_1, \dots, S_{|P|+|Q|+4})$ is indeed~$\pi'$.

Let $Y = \Xup_{P,Q} \cup \Xdo_{P,Q} \cup \{v_A, v_B\} \cup \Lup \cup \Ldo \cup \{t\}$
We show that for each vertex $\overline{z}$ in $\Lup \cup \{t\}$ (resp., $\Ldo$), there is a unique independent set $Y^{\overline{z}}$ of $Y$ such that $\overline{z} \in Y^{\overline{z}}$, and $\omega'(Y^{\overline{z}})$ is more than $\omega'(S \cap Y)$ (resp., $\omega'(t)$).
Note that $\omega'(S \cap Y) = \alpha (4n + 2) + \sum_{p \in P} \omega(vp)$.
Consider the following three cases:

\textbf{Case 1}: \textit{$\overline{z} = \lup_i$ for some $i \in [1,|P|+1]$.}
Then $\omega'(\overline{z}) = 2\alpha + \sum_{j < i} \omega(vp_i) + i$.
Further, the non-neighbors of $\overline{z}$ in $Y$ are exactly $v_A$ and $x_{v,p_j}$ for $j \geq i$.
Since $\omega(e) \geq 2n$ for all edges $e \in E$, if $Y^{\overline{z}}$ does not contain any of these non-neighbors of $\overline{z}$, then its total weight would be less than $\omega'(S \cap Y)$.
On the other hand, if $Y^{\overline{z}}$ contains all these non-neighbors, then its weight would be at least $\omega'(S \cap Y)$.
Hence, the uniqueness of $Y^{\overline{z}}$ holds in this case.

\textbf{Case 2}: \textit{$\overline{z} = t$.}
Note that $t$ is adjacent to all vertices in $Y$, so $S$ cannot have any other vertex in $Y$.
Note that $\omega'(t) = \omega'(\ldo_{|Q|+1}) + \omega'(v_B) - 1 = \alpha(4n+2) + \sum_{q \in Q} \omega({v,q}) - |Q| - 2$.
Since $\sum_{q \in Q} \omega({v,q}) > \sum_{p \in P} \omega({v,p})$, and since $\omega(e)$ is a multiple of $2n$ for all edges $e \in E$, it follows that $\omega'(t) > \omega'(S \cap Y)$.
The uniqueness of $Y^{\overline{z}}$ then follows.

\textbf{Case 3}: \textit{$\overline{z} = \ldo_i$ for some $i \in [1,|Q|+1]$.}
Then $\omega'(\overline{z}) = 2\alpha + \sum_{j < i} \omega(vq_i) - i$.
Further, the non-neighbors of $\overline{z}$ in $Y$ are exactly $v_B$ and $x_{v,q_j}$ for $j \geq i$.
By a similar argument as in Case~1, the only possibility for $Y^{\overline{z}}$ to have weights more than $\omega'(S'_{|P|+2} \cap Y) = \omega'(t)$ is that it contains all these non-neighbors.
Hence, we also have the uniqueness of $Y^{\overline{z}}$ in this case.

By the fact that $\pi$ is an improving sequence from $S$ and a simple induction, we can characterize $Y^{L[i]}$ as $S_i \cap Y$ for $i \in [1,|P|+|Q|+3]$.
The remainder of the proof goes as follows.
Firstly, we prove by induction that (*) for $i \in [1,|P|+|Q|+3]$, $S_i \equiv S'_i$.
Lastly, we prove that (**) there is a unique 3-swap to improve $S_{|P|+|Q|+3}$.
Note that (*), (**), and the fact that $\pi$ is an improving sequence imply that $\pi$ is a prefix of $\pi'$.
Since $\pi$ is a direct sequence, it means than $\pi = \pi'$, and the statement (b) of the claim then follows.

We now show (*).
For the base case $i =1$, recall that at the beginning, we argue the swap from $S$ to $S'_1$ adds a vertex $z$ in $(\Lup, t, \Ldo)$ and removes $v_A'$ and $v_A''$.
Since the vertices in $\{t\} \cup \Ldo$ are adjacent to $v_A \in S$, and since the vertices in $\Lup \setminus \{\lup_1\}$ are adjacent to a vertex in $\{x_{v,p \mid p \in P}\} \subset S$, it follows that $z$ must be $\lup_1$.
This implies $S'_1 = S_1$.

For the inductive case $i \in [2, |P|+|Q|+3]$, let $W$ be the improving 3-swap that transforms $S_{i-1}$ into $S_i$.
Suppose that $W$ adds a vertex $z' \notin Y$.
Note that for $w \in V \setminus \{v\}$, either $\{w_A, w_A', w_A''\} \subseteq S_{i-1}$ or $\{w_B, w_B', w_B''\} \subseteq S_{i-1}$.
Hence, $W$ cannot add any vertex in the two sets, since either the vertex is already in $S_{i-1}$ or the vertex is adjacent to three vertices in $S_{i-1}$.
Similarly, for a vertex of the form $x_{u,w}$, if $u \in B \setminus \{v\}$ or $w \in A \setminus \{v\}$, then $x_{u,w}$ is also adjacent to three vertices in $S_{i-1}$ (i.e., $w_A, w_A', w_A''$ or $u_B, u_B', u_B''$).
Further, if $u \in A \setminus \{v\}$ and $w \in B \setminus \{v\}$, then $x_{u,w}$ is in $S$ and hence also in $S_{i-1}$.
As $z' \notin Y$, all of the above imply that $z' \notin C$.
Hence, $z' \in D \setminus (\Lup \cup \Ldo \cup \{t\})$.

That means $z'$ is in a $A2B$- or $B2A$-simulator of a vertex $w$ (which may be $v$) for some partition $(P',Q')$.
If $w \neq v$, then $W$ needs to remove either $w_A'$ and $w_A''$ or $w_B'$ and $w_B''$.
Further, $W$ also needs to add $z'$ and removes the vertex in $S_{i-1} \cap (\Lup \cup \Ldo \cup \{t\})$.
Hence, $W$ cannot be a 3-swap, a contradiction.
If $w = v$, then we have $(P',Q') \neq (P,Q)$.
Hence,  it is easy to see that $M_{P',Q'} \setminus M_{P,Q}$ has at least three vertices, and these vertices must be in $S_{i-1}$ (since $p \in B$ for all $p \in P$, and $q \in A$ for all $q \in Q$ as argued above).
Because these vertices are adjacent to $z'$, $W$ needs to remove them, and hence $W$ cannot be a 3-swap, a contradiction.
Therefore, $W$ does not add any vertex outside of $Y$.

This implies that $\omega'(S_i \cap Y)$ is more than $\omega'(S_{i-1} \cap Y)$, which in turn is at least $\omega(S \cap Y)$ for $i \leq |P| + 2$ and $\omega'(S_i \cap Y) \geq \omega'(t)$ for $i > |P|+2$.
By the uniqueness and characterization of $Y^{z}$ above, and by the fact that $\omega'(S_{t+1} \cap Y) > \omega'(S_t \cap Y)$ for $t \in [1,|P|+|Q|+3]$, this implies that $S'_i \cap Y = S_j \cap Y$ for some $j \in [i, |P|+|Q|+3]$.
It can be easily checked that $j = i$ is the only value in this range such that $|(S_{i-1} \cap Y) \oplus (S_j \cap Y)| \leq 3$.
Further, $|(S_{i-1} \cap Y) \oplus (S_i \cap Y)| = 3$.
Hence, there can be no change outside of $Y$, and overall we have $S'_i = S_i$.

Finally, we show (**).
Suppose there is an improving 3-swap $W$ that transform $S_{|P|+|Q|+3}$ to an independent set $\overline{S}$.
Using the same argument as in the inductive case of (*), we note that $W$ cannot add a vertex outside $Y$.
If $W$ adds a vertex $z'$ in $Y \cap D$, then by the uniqueness and characterization of $Y^{z'}$, we have $\overline{S} \cap Y = S_j \cap Y$ for some $j \in [1,|P|+|Q|+3]$.
However, since $\pi$ is improving, we know that $\omega'(S_j \cap Y) \leq \omega'(S_{|P|+|Q|+3} \cap Y)$ for all $j \in [1,|P|+|Q|+3]$.
Since $W$ does not add any vertex outside $Y$, this implies $\omega'(\overline{S}) \leq \omega'(S_{|P|+|Q|+3})$, a contradiction to $W$ being improving.
Hence, $W$ can only add vertices in $Y \cap C$ and $\overline{S}$ contains no vertex in $D$.
Because of the biclique formed by $\{v_A, v_A', v_A''\}$ and $\{v_B, v_B', v_B''\}$ and the biclique by $\{x_{q,v} \mid q \in Q\}$ and $\{x_{v,p} \mid p \in P\}$, the only maximal independent sets of $Y \cap C$ are $\{v_A, v_A', v_A''\} \cup \{x_{v,p} \mid p \in P\}$ and $\{v_B, v_B', v_B''\} \cup \{x_{q,v} \mid q \in Q\}$.
Only the latter has weight more than $\omega'(S_{|P|+|Q|+3})$ and only by one.
This implies that $\overline{S} \cup Y = \{v_B, v_B', v_B''\} \cup \{x_{q,v} \mid q \in Q\}$.
In order for this to happen, $W$ has to remove $\ldo_1$ and add $v_B'$ and $v_B''$.
This means, $W$ does not change any vertex outside $Y$.
In other words, there is a unique way to improve $S_{|P|+|Q|+3}$.
This completes the proof of this claim.
\end{claimproof}

We now formally define the reduction.
Let $\plsredins$ be the function that maps an instance $I$ of \plsprob{Max Cut}{Flip} to the instance $I'$ of \plsprob{Weighted Independent Set}{3-Swap} as described above.
Note that for an independent set $S$ of $G'$ and for every vertex $v \in V$, at most one of the sets $\{v_A, v_A', v_A''\}$ and $\{v_B, v_B', v_B''\}$ has a nonempty intersection with $S$.
We define the solution mapping function $\plsredsol(I, S)$ as follows: For every vertex $v \in V$, if $\{v_A, v_A', v_A''\} \cap S \neq \emptyset$, then we add $v$ to $A$; otherwise we add $v$ to $B$. 
We then assign $(A,B)$ to be $\plsredsol(I, S)$.
It is easy to see that $\plsredins$ and $\plsredsol$ can be computed in polynomial time.
Combined with Claims~\ref{claim:R_all_lopt} and~\ref{claim:lopt_to_lopt}, this implies that $(\plsredins, \plsredsol)$ is a \pls-reduction.

Finally, we show that the set $R$ above satisfies all the conditions stipulated in \cref{def:tight_pls_complete}.
By \cref{claim:R_all_lopt}, $R$ contains all local optima, and hence it satisfies the first condition.
For every partition $(A,B)$ of $V$, we can construct $g(A,B)$ in polynomial time, and note that $\plsredsol(I, g(A,B)) = (A,B)$.
Therefore, the second condition is satisfied.
Lastly, \cref{claim:direct_sequence} implies the third condition of \cref{def:tight_pls_complete}.
Overall, this means that $(\plsredins, \plsredsol)$ is a tight \pls-reduction.
\end{proof}

We are now ready to prove \cref{thm:wis_all_exp}.

\begin{proof}[Proof of \cref{thm:wis_all_exp}]
Since a tight \pls-reduction preserves the all-exp property~\cite{schaffer1991}, \cref{thm:wis_all_exp} immediately follows from the \cref{thm:tight_to_wis}, the fact that \plsprob{Max Cut}{Flip} with maximum degree four and positive weights has the all-exp property~\cite{MonienT10}, and the observation that all weights of the \plsprob{Weighted Independent Set}{3-Swap} instance in the reduction are positive.
\end{proof}

\section{A new notion of reduction}
\label{sec:newred}
In this section, we define a new notion of \pls-reduction that transfers the result in \cref{thm:wis_hard} to other problems. 
We also demonstrate it by describing such a reduction from \pivotprob{Weighted Independent Set}{3-Swap} to \pivotprob{Max-Circuit}{Flip}, which is a special case of \pivotprob{Weighted Circuit}{Flip}.

\subsection{$\ell$-tight \pls-reduction}

\begin{definition}[$\ell$-tight \pls-reduction]
\label{def:l_tight_reduction}
	An \emph{$\ell$-tight \pls-reduction} from a problem $P \in \pls$ to a problem $Q \in \pls$ is a \pls-reduction $(\plsredins,\plsredsol)$ from $P$ to $Q$ such that for every instance $I \in D_P$, we can choose a subset $R$ of $F_Q(\plsredins(I))$ that satisfies:
	\begin{enumerate}
		\item $R$ contains all local optima of $\plsredins(I)$;
		\item For every solution $s_p \in F_{P}(I)$, we can construct in polynomial time a solution $s_q \in R$ so that $\plsredsol(I,s_q) = s_p$;
		\item If there is a path from $s_q \in R$ to $s'_q \in R$ in the transition graph $T_{\plsredins(I)}$ such that there are no other solutions in $R$ on the path, then either there is an edge from the solution $\plsredsol(I,s_q)$ to the solution $\plsredsol(I,s'_q)$ in the transition graph $T_I$ or these two solutions are identical;
		\item For two solutions $s_q \in R$ and $s'_q \in R$ such that either $\plsredsol(I,s_q) = \plsredsol(I,s'_q)$ or there is an edge $(\plsredsol(I,s_q),\plsredsol(I,s'_q))$ in the transition graph $T_I$, the shortest path from $s_q$ to $s'_q$ is at most $\ell$.
	\end{enumerate} 
\end{definition}

Note that the first three conditions are identical to the tight \pls-reduction in \cref{def:tight_pls_complete}.
While tight \pls-reduction forbids introduction of shortcuts in the transition graph, the distance between two solutions can be arbitrarily increased.
Our new condition 4 limits this increase.

\begin{theorem}
\label{thm:reduction}
	Let $\PP$ and $\PP'$ be two local search problems in \pls.
	Suppose there exists an $\varphi(\ell)$-tight \pls-reduction from $\PP$ to $\PP'$ for some integer $\ell$ and polynomial function $\varphi$.
	If there exists an algorithm to solve $\PP'[\textsc{Pivot}]$ in \FPT-time when parameterized by $\ell$, then there also exists such an algorithm to solve $\PP[\textsc{Pivot}]$.
\end{theorem}
\begin{proof}
	Suppose there exists an algorithm $A'$ that solves $\PP'[\textsc{Pivot}]$ in time $g(\ell) \cdot \poly{n}$ for some computable function $g$.	Let $(I, s)$ be an instance of $\PP[\textsc{Pivot}]$ (i.e., $I$ is an instance of $\PP$, $s$ a solution of $I$), and $\ell$ be the promised distance from $s$ to a local optimum.
	Let $(\plsredins,\plsredsol)$ be the $\varphi(\ell)$-tight \pls-reduction given in the lemma statement.
	Let $R$ be the set as guaranteed by \cref{def:l_tight_reduction} for the reduction $(\plsredins,\plsredsol)$ and instance $I$ of $\PP$.
	Note that since $(\plsredins,\plsredsol)$ is a \pls-reduction, in time polynomial in $|I|$, we can obtain the solution $I' := \plsredins(I)$ of $\PP'$. 
	By Condition 2 of \cref{def:l_tight_reduction}, in polynomial time, we can also obtain solution $s' \in R$ such that $\plsredsol(I,s') = s$.
	Note that if the promise is not violated, then by Condition 4 of \cref{def:l_tight_reduction}, there exists a path from $s'$ to a local optimum of $I'$ of length at most $\phi(\ell) \ell$.
	Then using the algorithm $A'$, we can find such a path in time $g(\varphi(\ell) \ell) \cdot \poly{|I'|}$.
	Note that the start and end of this path are in $R$, and the end of this path corresponds to a local optimum of $I$.
	Hence, applying Condition 3 of \cref{def:l_tight_reduction}, this corresponds to a path of length at most $\varphi(\ell) \ell$ from $s$ to a local optimum of $I$.
	In other words, we can solve $(I, s, \ell)$ in time $g'(\ell) \cdot \poly{|I|}$ for some computable function $g'$. 
\end{proof}

\begin{remark}
	\cref{claim:direct_sequence} implies that the 
 
reduction in the proof of \cref{thm:tight_to_wis} is an $\Oh(\Delta)$-tight \pls reduction.
\end{remark}

\subsection{Reduction to \plsprob{Max-Circuit}{Flip}}

Having formalized $\ell$-tight \pls-reductions and established their ability to transfer the desired lower bound, we proceed to our showcase of how this reduction can be applied to local search problems beyond \plsprob{Subset Weight Optimization Problem}{$c$-Swap}.

\begin{theorem}
	There is a $\Oh(c)$-tight \pls reduction from \plsprob{Subset Weight Optimization Problem}{$c$-Swap} to \plsprob{Max-Circuit}{Flip}.
\end{theorem}

\begin{proof}
	We adapt the proof of \pls-reduction from any problem in \pls to \plsprob{Min-Circuit}{Flip} in \cite{MR2450938}.
	Let $\mathcal{I} = (G, \omega, \sigma)$ be an instance of \plsprob{Subset Weight Optimization Problem}{$c$-Swap}.
	Let $G = (V, E)$ and $n := |V| + |E|$.
	We now construct a Boolean circuit $D$ as follows.
	Note that we can encode any solution of $\mathcal{I}$ in $n$ bits, where each bit is associated with a vertex/edge in $V \cup E$ and is 1 if and only if that vertex/edge is in the solution.
	For a length-$n$ string $u$ that encodes a solution $s$ of $\mathcal{I}$, let $\omega(u) := \omega(s)$, and let $B(u)$ be the set of strings encoding the solutions that can be obtained by an improving $c$-swap from $s$.
	For two strings $u$ and $w$, let $H(u,w)$ be the Hamming distance between $u$ and $w$, and let $V(u,w)$ be the set of $H(u,w)$ strings that are obtained if we change $u$ into $w$ by flipping the bits in which $u$ differs from $w$ in increasing order of their positions.
	We define a function $h : \{0,1\}^{V \cup E} \to \Q$ as follows.
	For a length-$n$ string $u$ that encodes a solution of $\mathcal{I}$, we define the value of $h$ for some \emph{structured} strings below:
	\begin{itemize}
		\item $h(uu00) = (2n+4) \omega(u)$,
		\item $h(uv00) = (2n+4) \omega(u) + H(u,v)$, where $v$ is a string in $V(u,w)$ for some $w$ in $B(u)$,
		\item $h(uw10) = (2n+4) \omega(u) + n + 1$, for $w \in B(u)$,
		\item $h(vu11) = (2n+4) \omega(u) - H(u,v) - 2$, where $v$ is a string in $V(w,u)$ for some $w$ such that $u \in B(w)$, and
		\item $h(uu10) = (2n+4) \omega(u) - 1$.
	\end{itemize}
	For an unstructured string $x$ (i.e., a string that is not of any form above), $h(x)$ has a value equal to $\mu$ minus the number of ones in $x$, where $\mu$ is smaller than any value for any structured string.
	Then the Boolean circuit $D$ is the one with $2n+2$ input nodes and computes the value $h(x)$ for a string $x$ of length $2n+2$.
	Such a circuit with polynomial size exists, as $h$ can be computed in polynomial time~\cite[Theorem 6.3]{MR2450938}.

	The description above constitutes the function $\plsredins$ that maps an instance $\mathcal{I}$ of \plsprob{Subset Weight Optimization Problem}{$c$-Swap} to an instance $D$ of \plsprob{Max-Circuit}{Flip}.
	The function $\plsredsol$ that maps a solution of $D$ to that of $\mathcal{I}$ is as follows: For a structured string $x$ of the form $uv00$ or $uv10$, $\plsredsol(\I, x) = u$ ($u$ and $v$ may be identical).
	For a structured string $x$ of the form $vu11$, we assign $\plsredsol(\I, x) = u$.
	Lastly, if $x$ is unstructured, $\plsredsol(\I, x) = 0 \dots 0$.
	
	It is easy to see that unstructured strings cannot be a local optimum.
	Further, among the structures strings, only the string of the form $uu00$ can be locally optimal, and this is if and only if $u$ encodes a local optimum of $\I$.
	Hence, $(\plsredins, \plsredsol)$ is a \pls-reduction.
	
	Now, consider the set $R$ as the set of all structured strings.
	We argue that the set $R$ satisfies the conditions in \cref{def:l_tight_reduction}.
	For the first condition, as argued above, only structured strings can be a local optimum of $D$, and hence $R$ contains all local optima.
	For the second condition, we map every solution $u$ of $\I$ to the solution $uu00$ in $R$.
	For the third condition, observe that an improvement step from a structured string can yield only a structured string.
	Hence, we only need to consider two adjacent structured strings $x_1$ and $x_2$ such that there is an edge from $x_1$ to $x_2$ in the transition graph $T_{D}$.
	By the choice of the function $\plsredsol$ and by construction, it follows that $\plsredsol(\I, x_1)$ and $\plsredsol(\I, x_2)$ are either identical or there is an edge from $\plsredsol(\I, x_1)$ to $\plsredsol(\I, x_2)$ in the transition graph $T_{\I}$.
	Hence, the third condition is satisfied.
	Lastly, consider a length-$n$ string $u$.
	The structured strings $x$ such that $\plsredsol(\I, x) = u$ either have the form $uv00$, $uv10$, or $vu11$ for some $v$ such that $H(u,v) \leq c$.
	Hence, for any two structured strings $x_1$ and $x_2$ such that either $\plsredsol(\I, x_1) = \plsredsol(\I, x_2)$, we must have $H(x_1, x_2) \leq 2c + 2$.
	Similarly, if there is an edge $(\plsredsol(\I, x_1), \plsredsol(\I, x_2))$ in $T_{\I}$, then we must have $H(x_1, x_2) \leq 4c + 4$.
	
	Overall, this shows that $(\plsredins, \plsredsol)$ is a $(4c+4)$-tight \pls-reduction.
\end{proof}

Combined with Theorems~\ref{thm:wis_hard} and \ref{thm:reduction}, the theorem above immediately implies the following.

\begin{corollary}
\label{cor:wis_hard_circuit}
	Unless $\FPT = \W[1]$, there does not exist an algorithm to solve \pivotprob{Max-Circuit}{Flip} in $g(\ell) \cdot  \poly{n}$ for any computable function $g$.
\end{corollary}

\section{Fixed-parameter Algorithms}
\label{sec:fptalgo}

In this section, we establish the algorithmic upper bounds that form the foundation of Main Finding~\ref{main:positive}. We recall that all results obtained in this section can also be directly translated to the \pls\ and standard algorithm problem formulations.

\subsection{\plsprob{Subset Weight Optimization Problem}{$c$-Swap}}
In this section, we show algorithmic results for problems in \pivotprob{Subset Weight Optimization Problem}{$c$-Swap} parameterized by the shortest distance $\ell$, swap size $c$ and the number~$k$ of distinct weights.
We start with an \XP algorithm parameterized by $\ell$ and $c$.

\begin{theorem}
	\label{thm:xp_subset}
	\pivotprob{Subset Weight Optimization Problem}{$c$-Swap} can be solved in time $\Oh(n^{c\ell})$, where $n$ is the instance size.
\end{theorem}
\begin{proof}
	Consider an instance $(G, \omega)$ of a problem in \pivotprob{Subset Weight Optimization Problem}{$c$-Swap}.
	Let $n := |V(G)| + |E(G)|$.
	By definition, any solution of $\I$ has at most $\Oh(n^c)$ neighbors in the transition graph.
	Hence, for any solution $Q$ of $\I$, there are at most $\Oh(n^{c\ell})$ solutions that are reachable from $Q$ via a path of the length at most $\ell$ in the transition graph.
	Exploring these solutions in a depth-first search fashion, we obtain an algorithm with the claimed running time.
\end{proof}

Next, we consider the case when both the swap size and the number of distinct weights are bounded in the following theorem.
We start with a straightforward \XP algorithm for this parameter following a simple observation below.

\begin{observation}
\label{ob:fpt_longest}
	For an instance $I$ of a problem $\PP \in \pls$ and a solution $s$ of $I$, let $\ell_{\max}$ be the longest distance from $s$ to a local optimum.
	Then the instance $(I,s)$ of $\PP[\textsc{Pivot}]$ can be solved in $\ell_{\max} \cdot \poly{n}$. 
\end{observation}
\begin{proof}
	Since $\PP$ is in \pls, for any solution $s'$, we can conclude that it is locally optimal or find a better solution in polynomial time.
	Since any improving sequence from $s$ has length at most $\ell_{\max}$, the statement then follows.
\end{proof}

\begin{corollary}
	\pivotprob{Subset Weight Optimization Problem}{$c$-Swap} can be solved in time $\Oh(n^{2k + c})$, where $n$ is the number of vertices and~$k$ is the number of distinct weights.\end{corollary}
\begin{proof}
	It is easy to see that given a set $R$ of $\Oh(n^2)$ numbers among $k$ different choices of values, the set $\{ \sum{s \in S \mid S \subseteq R}\}$ has $\Oh(n^{2k})$ elements.
	This implies that there are $\Oh(n^{2k})$ possible objective values for the \pivotprob{Subset Weight Optimization Problem}{$c$-Swap} instance.
	Since all solutions in an improving sequence must have pairwise distinct objective values, it follows that the length of the sequence is $\Oh(n^{2k})$.
	Combined with \cref{ob:fpt_longest}, the corollary then follows by the fact that it takes $\Oh(n^c)$ time to find a better solution in the local neighborhood, if one exists.
\end{proof}

Next, we show an \FPT-algorithm (\cref{cor:fpt_diff_weights} below) based on the following result, which already had a big impact on kernelization results for weighted optimization problems~\cite{EKMR17}.

\begin{theorem}[Frank and Tardos~\cite{MR905151}]
\label{thm:frank_tardos}
	Given a rational vector $w \in \Q^k$ and an integer $N$, there is a polynomial-time algorithm to output an integral vector $\overline{w} \in \Z^k$ such that $|\overline{w}|_{\infty} \leq 2^{4k^3} N^{k(k+2)}$ and $\sign(w^{\top} b) = \sign(\overline{w}^{\top} b)$ for any integer vector $b$ with $|b|_1 \leq N - 1$.
\end{theorem}

\begin{theorem}
	\label{thm:diff_weights}
	For \plsprob{Subset Weight Optimization Problem}{$c$-Swap} with $k$ different values of weights and for any pivoting rule, the standard local search algorithm takes $g(c,k) \cdot n^{\Oh(1)}$ steps for some computable function $g$. 
\end{theorem}
\begin{proof}	Applying \cref{thm:frank_tardos} with $N = c + 1$ and $w = (w_1, \dots w_k)$ being the vector of the $k$ different weights, we obtain a vector $\overline{w} = (\overline{w}_1, \dots \overline{w}_k)$ such that (i) $|\overline{w}|_{\infty} \leq g(c,k)$ for some computable function $g$ and (ii) $\sign(w^{\top} b) = \sign(\overline{w}^{\top} b)$ for any integer vector $b$ with $|b|_1 \leq c$.
	Let $\I$ be the original instance and $\overline{\I}$ be the instance obtained from $\I$ by replacing any weight with value $w_i$ by $\overline{w}_i$ for $i \in [k]$.
	(i) implies that the absolute value of the objective for $\overline{\I}$ is bounded by $n^{\Oh(1)} g(c,k)$.
	Since the weights are integral, the objective improves by at least 1 at each step.
	Hence, the standard local search algorithm for $\overline{\I}$ should terminate in $n^{\Oh(1)} g(c,k)$ steps for any pivoting rule.
	Next, (ii) implies that the transition graph is the same for both $\I$ and $\overline{\I}$.
	As such, a local optimum for $\overline{\I}$ is also one for $\I$.
	The theorem then follows.
\end{proof}

Since we can find an improving swap in $n^{\Oh(c)}$~time, \cref{thm:diff_weights} immediately implies the following FPT result.

\begin{corollary}
\label{cor:fpt_diff_weights}
	\pivotprob{Subset Weight Optimization Problem}{$c$-Swap} with $k$ different values of weights can be solved in time $g(c,k) \cdot n^{\Oh(c)}$ for some computable function $g$.
\end{corollary}

Note that we cannot presumably replace the~$\Oh(c)$ in the exponent by a constant.
This is due to the fact that it is W[1]-hard~\cite{FellowsFLRSV12} for parameter~$c$ to decide whether a given independent set in an unweighted graph is locally optimal. 
Moreover, as a consequence from the reduction behind~\cite[Theorem 3.7]{KomusiewiczM22}, the problem cannot be solved in $g(c) \cdot n^{o(c)}$~time, unless the ETH fails.

\subsection{\plsprob{Generalized-Circuit}{Flip}}
The same proof techniques in the previous section can also be applied to \pivotprob{Generalized-Circuit}{Flip}.
First, we obtain a similar \XP result when parameterizing the problem by the shortest distance $\ell$.

\begin{theorem}
	\label{thm:xp_circuit}
	An instance of \pivotprob{Generalized-Circuit}{Flip} with $n$ input gates can be solved in time $\Oh(n^{\ell})$.
\end{theorem}
\begin{proof}
	Every solution has at most $n$ neighbors in the transition graph.
	Hence, exhaustive search of all solutions reachable from the initial solution by a directed path of length at most $\ell$ in the transition graph runs in time $\Oh(n^{\ell})$.
\end{proof}

Note that there are only $2^m$ states of the output, and in any directed path in the transition graph, we cannot encounter any state twice.
This implies that the distance from a solution to any reachable local optimum is at most $2^m$ (i.e., $\ell < 2^m$).
While the result above then implies an \XP algorithm when parameterizing by $m$, we can instead immediately obtain an \FPT-algorithm from \cref{ob:fpt_longest}.

\begin{corollary}
	An instance of \pivotprob{Generalized-Circuit}{Flip} with $n$ input gates and $m$ output gates can be solved in time $2^m \poly{n}$.
\end{corollary}

If we parameterize by the number $k$ of distinct weights and an additional parameter $t$ defined in \cref{thm:diff_weights_circuit} below, we again obtain that any pivoting rule yields a fixed-parameter algorithm.

\begin{theorem}
\label{thm:diff_weights_circuit}
Let $\I$ be an instance of \pivotprob{Generalized-Circuit}{Flip} with independently chosen weights among $k$ different values.
Further, let $t$ be the maximum number of output gates connected to any particular input gate. 
Then for any pivoting rule, the standard local search algorithm takes $f(t,k) \cdot n^{\Oh(1)}$ steps.
\end{theorem}
\begin{proof}
	The proof is analogous to that of \cref{thm:diff_weights}, when we replace $c$ by $t$.
\end{proof}

\section{Concluding Remarks}
\label{sec:concl}

There are numerous avenues for future research. First, one could identify further local search problems for which the pivoting formulation parameterized by the distance~$\ell$ of the starting solution to the closest local optimum is not \FPT.
Conversely, it is open whether there is any \emph{natural} \pls-complete problem whose pivoting formulation admits a fixed-parameter algorithm when parameterized by~$\ell$. 
Observe that such a fixed-parameter algorithm exists for \pls-complete local search problems where the largest neighborhood has constant size or---slightly more generally---the maximum outdegree in the transition graph is constant. 
Such problems do exist, for example one may define an ordering of solutions and redefine the neighborhood to only contain the smallest solution with respect to this ordering.  We are, however, not aware of any natural \pls-complete problem with only constant-size neighborhoods. 
In light of this discussion it seems more meaningful to ask whether there is any \pls-complete problem with an unbounded number of large neighborhoods whose pivoting formulation is \FPT~with respect to~$\ell$.

In addition to~$\ell$, other parameters related to the transition graph could be considered. 
For example, since parameterization by the diameter of the transition graph or the distance to the furthest local optimum trivially yield fixed-parameter algorithms, one could consider parameters that are sandwiched between~$\ell$ and the diameter. By the discussion above, it is also motivated to consider parameter combinations~$(q,\ell)$ where~$q$ is some parameter that is always upper-bounded by the maximum out-degree of the transition graph.

Finally, it is open to obtain any parameterized hardness results for classic (non-pivoting) formulations of \pls~problems. A major obstacle towards such a result is that the existence of any polynomial-time Turing reduction of some \NP-hard problem~$X$ to a problem~$L$ in \pls implies \NP=\coNP~\cite{JPY1988}. This excludes the standard approach of providing a polynomial parameter transformation from \textsc{Multicolored Independent Set} or similar classic problems. In other words, any parameterized reduction from~\textsc{Multicolored Independent Set} or similar problems would need an exponential running time dependence on the parameter~$k$. An alternative, perhaps more promising, approach would be to develop analogs of \XP or \W[1] for the \pls~class and to identify complete problems for such classes. This, however, would not relate the hardness of such problems to existing complexity classes.

\bibliography{refs}

\end{document}